\documentclass[11pt]{article}

\usepackage{graphicx} 
\usepackage{hyperref}
\usepackage{xcolor}
\hypersetup{
    colorlinks,
    linkcolor={red!50!black},
    citecolor={blue!50!black},
    urlcolor={blue!80!black}
}

\begin{document}

\newtheorem{theorem}{Theorem}[section]
\newtheorem{lemma}[theorem]{Lemma}
\newtheorem{corollary}[theorem]{Corollary}
\newenvironment{proof}[1][Proof]{\begin{trivlist}
\item[\hskip \labelsep {\bfseries #1}]}{\end{trivlist}}

\newcommand{\qed}{\nobreak \ifvmode \relax \else
    \ifdim\lastskip<1.5em \hskip-\lastskip
    \hskip1.5em plus0em minus0.5em \fi \nobreak
    \vrule height0.67em width0.67em depth0.0em\fi}

\title{Turing Tumble is Turing-Complete}  

\author{Lenny Pitt\\
Department of Computer Science\\
University of Illinois\\
Urbana, IL\\
lennypitt@gmail.com}

\maketitle

\begin{abstract}
It is shown that the toy Turing Tumble, suitably extended with an infinitely long game board and unlimited supply of pieces, is Turing-Complete.  This is achieved via direct simulation of a Turing machine.  Unlike previously informally presented constructions, we do not encode the finite control infinitely many times, we need only one trigger/ball-hopper pair, and we prove our construction correct.    We believe this is the first natural extension of a marble-based computer that has been shown to be universal.  
\end{abstract}

\section{Introduction}

{\href{https://en.wikipedia.org/wiki/Turing_Tumble}{Turing
Tumble}}  is a mechanical toy introduced in 2018 after a very successful
kickstarter campaign. As highlighted
on {\href{https://www.turingtumble.com}{the
website}},  the toy allows the simulation of basic logic gates via
plastic components laid out on a ramped board and driven by small metal
marbles that roll down the board through the component layout to affect
various simple computations. The next section describes the  available pieces and 
their functions.  

\begin{figure*}[htp]
\begin{center}
\includegraphics[width=1.5in]{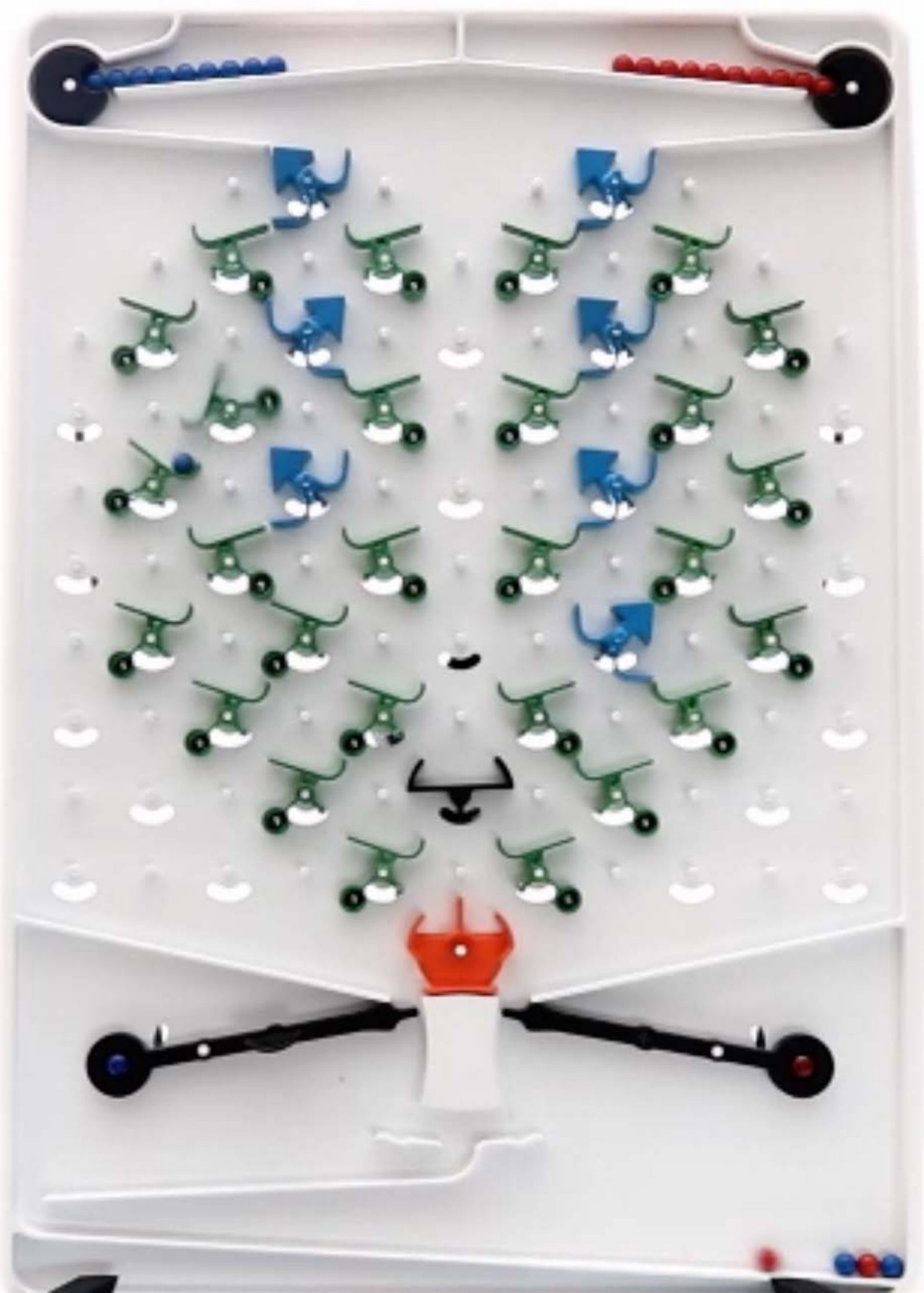}\\
 A Turing Tumble layout that can perform addition.\\
 (Image from 
 {\href 
 {https://turingtumble.com}{turingtumble.com}})
\end{center}
\end{figure*}

The toy is in the spirit of \href{https://en.wikipedia.org/wiki/Digi-Comp_II}{Digi-Comp II}
and {\href{https://en.wikipedia.org/wiki/Dr._Nim}{Dr.Nim}}, but is seemingly much more powerful.  It has won awards from
toymakers and from parents.  Besides all of the children who are
learning about computation via this device, it has a fairly robust
online Turing Tumble Community
({\href{https://community.turingtumble.com}{TTC}}) 
of aficionados who enjoy sharing challenging new puzzles, solutions,
educational tips, new techniques, maker-advice, and proposals for new
pieces. There has been one popular thread focused on conversations
about the power of the toy.
A {\href{https://tinyurl.com/ttsimulator}{simulator}} has
been built that allows experimenting with layouts larger than the
provided board, and use of more (though not different) pieces than are
provided with the toy. The reader unfamiliar with how the device works
is encouraged to read Section~\ref{sec:parts} first, and to experiment with the different pieces using the
simulator, or a purchased game, so that the narrative here will be
easier to follow. We present some involved constructions, and having an
intuition about how the pieces work will be very helpful.

\subsection*{Prior Work}

Most of the work on devices that compute using marbles has focused on
computational complexity as opposed to computability.  Aaronson~\cite{aaronson2014}
discusses, based on complexity assumptions, why the Digicomp-II is not
likely to be ``circuit-universal'', i.e., powerful enough to program any
given circuit C, using sufficiently many pieces and a sufficiently large
board to realize C's computation.  In his blog, he discusses the power
of the Digicomp II, showing that a natural extension is CC-complete (CC
= ``comparator circuit''), placing it somewhere between the complexity
classes NL and P, and showing that, unless CC=P, it is not
circuit-universal.  He discusses the property of being ``1-Lipschitz'',
meaning one input bit change can only affect one output bit - as a
disqualifying feature for a universal computer. Digicomp II has this
property. Turing Tumble does not.

Bickford~\cite{bickford2014} gave a presentation at G4G11 and subsequently made a blog post
in which he discusses marble computations that use
switches only, along with marble runs that connect the switches. The
marble paths appear to allow cycles, thus precluding a gravity-powered
realization.  The switches are essentially the bits of Turing Tumble,
Digicomp II, and of Dr. Nim - 2-state components through which a marble
may be routed, and which toggle the state.  He shows that a
``double-switch'' machine, which has pairs of switches ganged together
so that their states are always equal, has a halting problem that is
PSPACE-Complete. This is done by constructing a polynomial space TM out
of only double-switches and (not necessarily planar nor acyclic)
connections. He does not address an extended model that allows for
arbitrary unbounded computation.

More recently, Johnson~\cite{johnson2019} investigates the complexity of TT itself,
and shows that the TT decision problem (do any marbles ever reach a
particular interceptor) is P-Complete via the Circuit Value Problem, and
PSpace complete when an exponential number of marbles is allowed.
~Again, there is not an extension to universal computation, but instead
to a larger finite board and piece supply.   Johnson also discusses
some models that are only very loosely related to marble computations,
but are of interest, including Demaine et al.'s discussion of
motion-planning devices~\cite{dglr2018}, billiard ball - based computations~\cite{margolus1984},
particle computation models~\cite{bdflm2019}, and tilting-labyrinth type games.

In order to understand previous attempts to show that TT is universal
though, and to understand the contribution of this paper, we must first
understand what that means.\\

{\em What is an acceptable ``unbounded'' extension of the toy?}\\

A TM has an infinite tape.  In contrast, TT has a finite board and
comes with finitely many pieces, so cannot be ``universal'' as it can
only perform computations bounded by a function of the size of the board
and the number of pieces used.  Of course, that is not what is meant by
any type of completeness when discussing finite toys and games. Whether
discussing computability or complexity based completeness, what is
typically investigated are ``reasonable extensions'' of the toy or game
in question, which allow a growing size of board or number of pieces.
~The question then is to what extent these extended models increase in
complexity or power, as the board and number of pieces grow.

In this case, we ask specifically if one adds a ``reasonable''
unbounded or infinite component to Turing Tumble still in the spirit of
the toy, can a TM be simulated, or, equivalently, can any computation
(not just one bounded by a finite parameter) be otherwise achieved? Can
any computation be performed, assuming sufficient playing board space
and pieces are available?

This begs the question of what is ``reasonable'' and ``in the spirit of
the toy''?

As mentioned, requiring a fixed, finite number of TT pieces results in
a device no more powerful than a finite-state machine. However, if we
allow a different finite number of pieces for each input size, as we do
with circuits, then every finite-sized computer can be simulated as
shown in Section~\ref{sec:finite-computer}, and as shown previously with interesting complexity results~\cite{johnson2019}.

It seems that at the least, allowing an infinite board (to ensure that
the toy is not computationally space-bounded, and an infinite number of
pieces (to ensure that its state-space is not bounded) will be required,
and do not seem unnatural.

However, triggers/ball-hopper pairs are somewhat different from the
other pieces, in the following ways:
\begin{itemize}
\item
  {They serve only to (re)start computation}
\item
  {They are not typically ``locally connected'' to adjacent pieces on the board.}
\item
  There is a potentially long-distance physical connection on the back of the board between each trigger and its corresponding ball-hopper.  This
  is so the connecting piece does not interfere with the rest of the
  components on the front of the board. 
\end{itemize}
Consequently, realizing great distances between triggers and
corresponding ball-hoppers, or multiple triggers connected to the same
ball-hopper, or infinite pairs of trigger and ball-hoppers, might be
difficult to implement physically, whereas it is easy to imagine an
extended board with additional pieces of the other types.  Without
triggers, one could simply use a finite sized board and number of
pieces, and add more as needed by the computation.  Triggers though might
 require connecting each new board back to one or more earlier boards.  
Thus, whether it is reasonable to allow an infinite number of
trigger/ball-hopper pairs, or even allow an infinite number of
triggers connected to a single ball-hopper, is perhaps in the eye of the
beholder.  

Our construction of Section~\ref{sec:simulation} provides a simulation of a TM with three ball-hoppers, but with an infinite number of triggers paired
with each hopper.  This is eliminated in Section~\ref{sec:improvements}, where we present an extension using only a single trigger/ball-hopper pair by relying on infinite gear chains for pieces to communicate long distance. We explain how the infinite chains may be eliminated at the expense
of computational time.

Another question that arises in any simulation of a computation is
whether the encoding is doing the work of the computation.  A safeguard
against this situation is to ensure that, except for finitely many
pieces (e.g., those encoding the input), the remaining pieces are
initialized to a simple, finite, repeating pattern, that is independent
of the input.  (For example, a proof that
 {\href{https://en.wikipedia.org/wiki/Rule_110}{Rule
110 for cellular automata}} is Turing-complete requires a repeating
infinite pattern to be initially written in the cells of the cellular
automaton, rather than just an initial finite-length pattern. That
infinite pattern forms the background milieu within which the actual
simulation is carried out.)

Because TT is circuit-universal, and a TM's finite control can be
realized in a simple circuit, a construction to simulate a TM is
possible provided there is a way to connect TT ``circuits'' together into an infinite linear array to simulate a TM tape - with each circuit acting as a tape cell -  with a mechanism for interaction between the adjacent circuits to allow the TM to move right or left by ``activating'' the adjacent circuit, and to carry state information.

This is what T. Yamada (user ``Yama-chan'') presents in a \href{https://community.turingtumble.com/t/proof-of-turing-completeness/372/46}{post on TTC} ~(2018). A TM computation is simulated by an infinite horizontal array of TT boards. Each board has the entire TM finite control encoded on it, and has triggers to ball hoppers to the boards on the left and on the right. There are gear chains between adjacent
boards to carry state variable information. The construction therefore
does have an infinitely repeating pattern in all but finitely many
cells, but that pattern can be as complex as the TM's transition
function.

A paper by Tomita et al.~\cite{tlipyk2019} introduces a formal model of the toy,
called Turing Tumble Model (TTM), and shows how to simulate (a finite structure built from) Reversible
Logic Elements (RLE) in the model. Because chaining such RLEs together in
an infinite one-dimensional array provides a universal model of
computation, they claim that this is sufficient to show that TTM is
computationally universal. However, extending the simulation of RLEs by
TTMs to a simulation of an {\em infinite chain} of RLEs by an {\em infinite chain} of TTMs is not presented.   In private communication the authors indicate that the issues arising from an infinite extension may be adequately dealt with by connecting the TTMs together locally to the left and right via a {\em feedback system} in TTM.   We believe this would translate to trigger/hopper connections between adjacent boards for Turing Tumble.   It appears that this construction is comparable to that proposed by Yamada above  (i.e., the local TTMs carry the transition function and state information,  communication between cells is via gear chains, with head movement simulated by trigger/hopper pairs across adjacent cells.)

In an {\href{https://docs.google.com/document/d/1HmGJwkIMCrnXdrB20eZuiuXYMDz4DddypYoSYGbWLiM}
{earlier version of this paper}  we improve the
results of Yamada and Tomita et al., to use only a single trigger/hopper pair, but still have to
encode the transition function of the TM as part of the infinitely
repeating pattern of tape cells - which we believe is unnatural.

Our initial construction presented in Section~\ref{sec:simulation} uses infinitely many triggers for three ball hoppers, but improves previous results because it encodes the TM just once as opposed to in each of the infinitely many cells of the simulated tape.   A modification described in Section~\ref{sec:improvements}  allows for just a single trigger (and thus only one ball hopper). 

More specifically, the board is infinite vertically and bounded
horizontally.  The only information about a computation that is
encoded is the finite control of the TM (at the top), the finite input,
and an infinite repeating pattern of connected 0 bits representing the
empty rest of tape, along with a handful of other variables.   This
repeating pattern is independent of the TM being simulated, and of the
input.  Thus, the computation itself is not somehow being encoded.  Moreover, the simulation is structured very much like a TM: a finite
control working on a separate infinite read/write tape.

Shortly after the toy was introduced in 2018, an interesting discussion
commenced on the {\href{https://community.turingtumble.com/}{Bboard}}
about what it meant to say the toy was ``Turing-complete'', as was
stated by the manufacturer. Shortly thereafter, 
\href{https://community.turingtumble.com/t/proof-of-turing-completeness/372/29}{a simulation of cellular
automata} by J. Crossen using TT was presented, and subsequently over
the last three years there has been a lot of discussion and additional
constructions shown on the Bboard~\cite{ttc}. Because cellular
automata include those capable of executing \href{https://en.wikipedia.org/wiki/Rule_110}{``Rule
110''} which itself can simulate computation (via a chain of other
reductions), assuming that the cellular automaton TT implementation is
correct, this shows that TT is Turing-complete. 

Unlike the simulations of cellular automata that have been presented,
the constructions presented here have no balls ``falling to infinity'', nor does our main constructions employ multiple balls falling at the same time, which would make reasoning about its behavior more challenging. (However, our stronger result in Section~\ref{sec:improvements} does allow multiple balls to fall at once in a very prescribed way, but we argue that they cannot interfere with each other.)
  
We believe that the previous constructions of Yamada, Tomita et al., and Crossen, probably show that Turing Tumble is Turing-complete, although none have been proven correct nor undergone formal review.  Moreover, none are ideal in that the construction may be incomplete or lack detail, and/or may employ an infinite number of trigger/hopper pairs, 
and/or may encode the TM at every simulated cell.

The  reason for the present work then, besides the primary one of having fun, was to formally present for review a new construction directly simulating a Turing machine with Turing Tumble, and to provide a proof that the construction is correct.

We believe that Turing Tumble may be the first marble-based computing device that is inherently powerful enough to admit a proof of
universality in a naturally extended model of the finite machine.

\section{Preliminaries}

\subsection{Turing Tumble Parts}
\label{sec:parts}

The following shows the parts available in Turing Tumble. 
Also available are \href{https://drive.google.com/drive/folders/11O_oCy-xpZSBxyFfnYTahPYoKvrV7t0n?usp=sharing}{short animated gifs of the parts in action}, courtesy of the folks at 
{\href{https://turingtumble.com}{Turing Tumble}}.


\begin{description}
\item  \raisebox{-.125in}{\includegraphics[scale=.5]{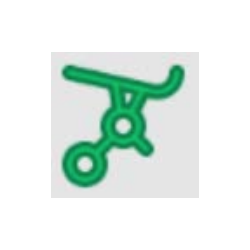}}
 A ramp.  A ball enters from either direction from above, and this moves the ball down and to the right. Its mirror image moves the ball down and to the left.
This is used to move balls down the board by sequencing these into
ramp ``chains'' - either diagonally, or vertically by using
opposite-facing ramps in a column.

\item  \raisebox{-.175in}{ \includegraphics[scale=.6] {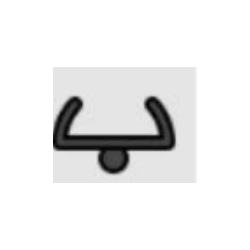}}
 An interceptor catches a ball, ending (perhaps part of) a ``computation''. 

\item  \raisebox{-.125in}{\includegraphics[scale=.5]{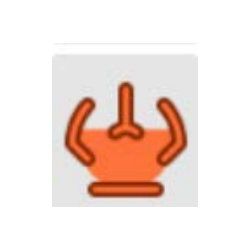} } A crossover. A ball entering from
the upper left exits to the lower right. A ball entering from the upper
right exits to the lower left. This is typically used to allow balls
heading different directions along ramp chains to cross over other ramp
chains, to allow otherwise non-planar layouts to be used on the board. 

\item
 \raisebox{-.125in}{\includegraphics[scale=.5]{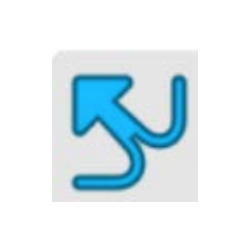} } A bit.   This has two states:
~pointing left (shown) representing ``0'', or pointing right (mirror
image) representing ``1'' (not shown).  If a ball enters from above
while the gear bit is pointing left (has value 0), it flips the bit to
the right (has value 1) and exits to the lower right. Conversely, if a
ball enters while the gear bit is pointing right, it flips the bit to
the left and exits to the lower left.  This is the basic unit of
information storage, and I/O, and can be used to create two branching
paths exiting from one ramp coming in, thus allowing a 0-1 branch
instruction. Usually ramps and a crossover are placed below, so that if
the bit were set to 0, the ball will be routed to the lower left
instead of the lower right, and conversely, if the bit were set to 1,
it will exit to the lower right.

\item  \raisebox{-.125in}{\includegraphics[scale=.5]{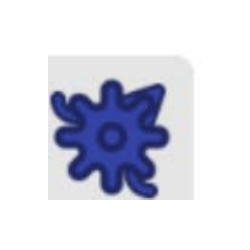}}  A gear bit. These behave like
the standard bit (just above) with a 0 state pointing left, and a 1
state pointing right (shown). Used alone, gear bits behave just like
standard bits, However, they can also behave like gears, and have great
advantage when they are connected to other gear bits via gear pieces
(shown just below).   One of the problems with a standard bit piece is
that when it is read (i.e., when a ball passes by), its state changes.
~By combining gear bits as shown in the next section, we can create
small substructures that allow a read while resetting the bit to its
original orientation.   Or, we can create a
``write'' instruction, that will always exit with the bit pointing in a
desired direction, regardless of how it was pointing initially.

\item \raisebox{-.125in}{ \includegraphics[scale=.5]{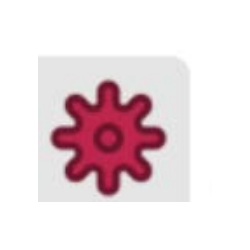} }  A gear. Behaves like a gear.  Used to connect gear bits.

\item \raisebox{-.15in}{\includegraphics[scale=0.45]{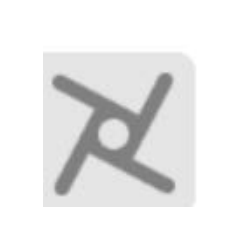} }and
\raisebox{-.1in}{\includegraphics[scale=0.2]{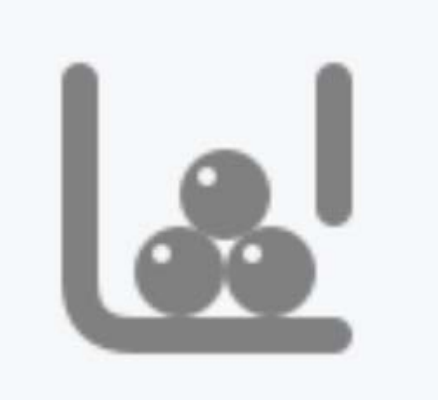}}  A trigger and a ball-hopper. 
When a ball passes through the trigger, another ball is released from
the ball hopper.   In the actual TT toy, there are two triggers at the
bottom of the board, which look like giant levers (not as depicted
here). The different triggers cause balls to be released from one of
two ball hoppers at the top of the board, typically holding blue balls
(left side) or red balls (right side).  The icon shown here is from the
online simulator, which allows the placement of triggers anywhere on the
board.  Designing a real-world piece that could actually be placed
(and removed and replaced) anywhere and that could cause a ball to be
released from the top seems quite difficult because a physical
connection between the piece and the ball release must be created. (The
trigger levers in the TT game connect to the ball release mechanisms via
long permanent connectors behind the board. This way, the connectors do
not interfere with the rest of the pieces placed on the board.)  In
the simulator, multiple triggers can be associated with a single
ball-hopper.
\end{description}

\subsection{Useful Components}

{We rely heavily on the following structures built from gear bits:}

\begin{description}

\item  {\includegraphics[scale=0.25]{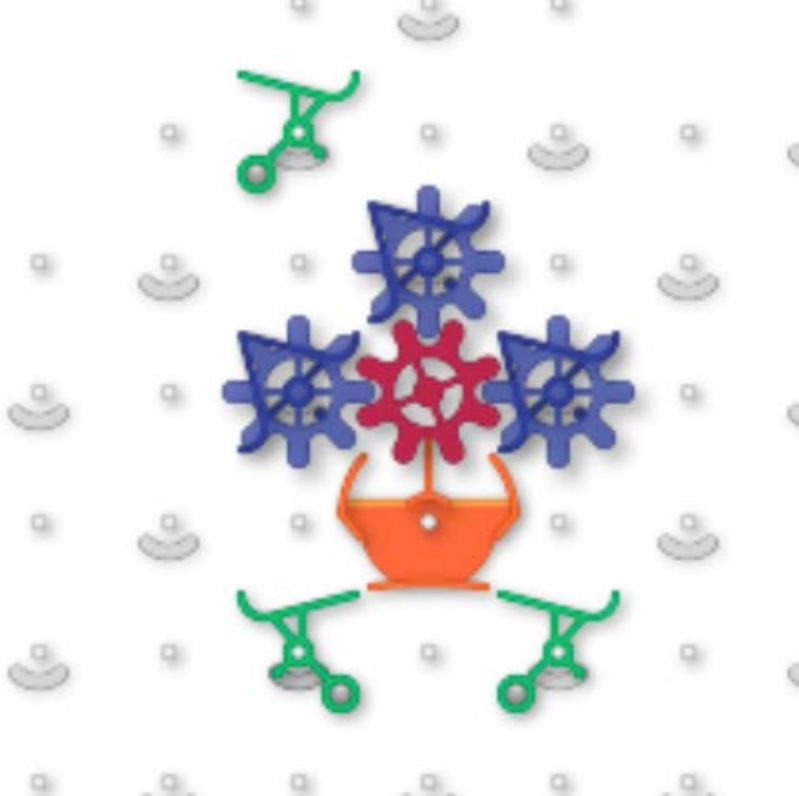}}
  {(}{N}{on-destructive) read.  This construct allows the ball to pass
  through a gear bit, leaving the gear bit's state unchanged, and
  }{exiting}{~either toward the left, indicating that the gear bit was
  set to 0, or toward the right, indicating the gear bit was set to 1.
  The idea is that the three gear bits (purple) are linked together via
  a gear (red). A ball enters the top gear bit and based on the state,
  falls either to the right or left, changing the state (of all three)
  momentarily. Either way it falls, it encounters a second linked gear
  bit, which toggles the state (of all three) back to the original
  state, and then is routed into the crossover' directing the ball to
  the left or right, depending on whether the gear bit was in the 0
  state, or the 1 state, respectively.  The location of the red gear is
  irrelevant, as long as the three gear bits are linked.}

\item {\includegraphics[scale=0.25]{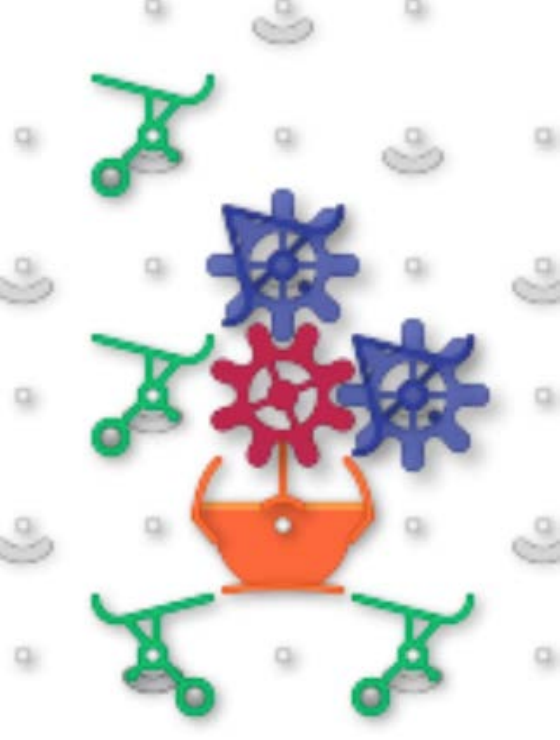}}
  Write-0.  This allows the ball to pass through a gear bit, setting
  it to 0, but exiting to the left if the gear bit was originally 0, and
  to the right if the gear bit was originally 1.  The idea here is
  that if the bit was set to 0, it falls to the right and it is toggled
  twice just as in the nondestructive read, so maintains the value 0.
  But if it was set to 1, then it falls to the left, and toggles the
  bits just once, thus setting the bits to 0.  (Again, the location of
  the red gear is irrelevant, as long as the three gear bits are linked.)

{In cases where a write is required, but no branching, one can simply
eliminate the crossover (orange) and have a single path stretching down
below the write component. }

\item {\includegraphics[scale=0.25]{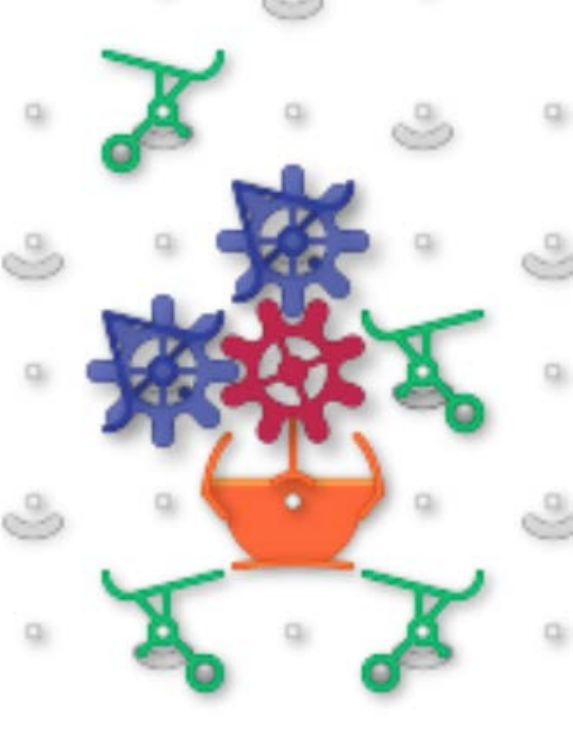}}
  Write-1.  This is the same as Write-0, but the gear bit is set to 1
  after the ball passes through.    As with writing a 0, if no branching is desired, we eliminate the
crossover and have a single exiting path regardless of the value of the
bit initially.

\item {\includegraphics[scale=0.25]{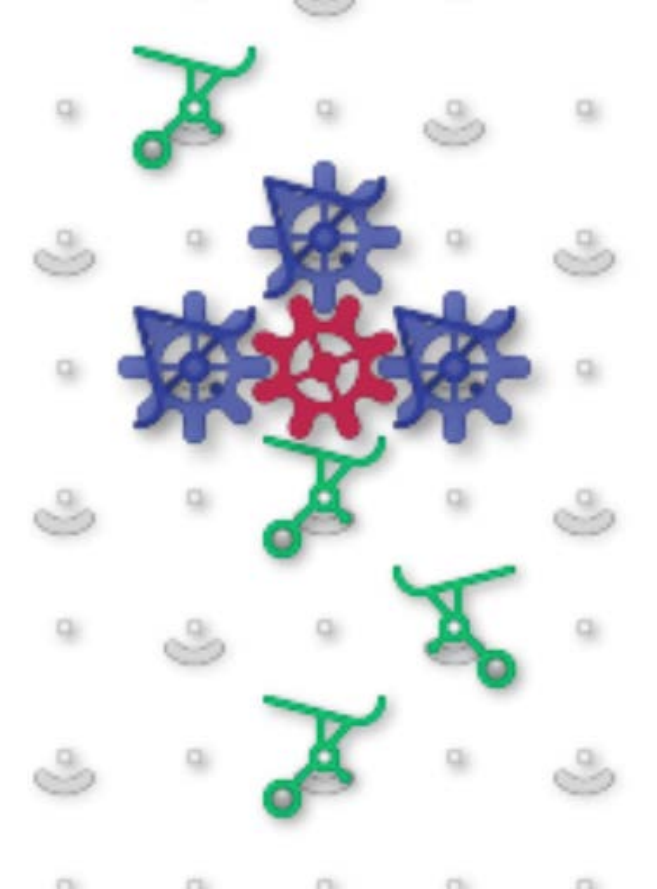}}
  Ignore.  This allows a ball going down a path of ramps (the green
  pieces) to pass through a gear bit along the path, leaving the gear
  bit unchanged, and exiting on a continuation of the path (no
  branching). This is just a nondestructive read, with no branching
  below.  (Again, the location of the red gear is irrelevant, as long
  as the three gear bits are linked.)
\end{description}

\section{Simulating the Computation of any Computer with Finite Resources}
\label{sec:finite-computer}

A typical approach in computing with TT is to represent variables as
bits on the TT board. Multiple bits can be laid out vertically, and a
ball falling down can pass through each bit, and branch off the chain
based on the bit value as just seen, in order to achieve some desired
effect.   Ideally, after passing through $n$ variables, we would
like the ball to be traversing down any one of up to $2^n$~vertical
paths.  However, in order to achieve this, we actually need multiple
copies of each variable, because after the first variable is read, and a
branch is executed, we have two distinct paths upon which the second
variable will need to be tested.

So, instead of representing the memory of the computer with a stack of
$n$ bits, we'll have lots of duplication, and represent it in a binary
tree structure with $n$ levels made up of gear bits. The $k$th level will
have $2^{k-1}$ gear bits, all geared together so that they always show
the same value. Hence, there are effectively only $n$ different functional
gear bits -- one for each level -- and exactly $2^n$ configurations of
these gear bits. (The total number of gear bits used to represent these
$n$ distinct values in a tree is $2^n-1$).  If the gear bits are
packed too closely, then there may be as few as $n+1$ exit points from the
tree, but if they are sufficiently far apart, then we can ensure that
there are exactly $2^n$ distinct exit points, and can also extend the
$2^n$~paths from the bottommost gear bits to ensure that the ball may
arrive at any one of $2^n$ distinct locations each separated horizontally by at
least, say, 10 spots, which will allow subsequent gates to be placed
below each path without interfering with gates placed below other paths.

Figure \ref{fig:branching-tree} below provides an example of a ``Read-Branch'' tree of height
$n=3$, with $2^3=8$ branches extending from the bottom of the tree. Due to
space constraints, we've only shown it with little separation between
the 8 branches at the bottom, and this wouldn't work with the rest of the
constructions needed in this paper. But by making the tree larger and
the levels farther apart (but the same number of levels $n = 3$), we could
have any separating width we'd like, because the width of such a tree
grows linearly with the height. 

Note that the ``reads'' that are done are destructive: the bit flips
as it is read. We could easily add some gears and additional gear bits
to make them non-destructive reads, as in the Preliminaries section just
above.
\begin{figure}[h]
\begin{center}
\includegraphics[width=4in]{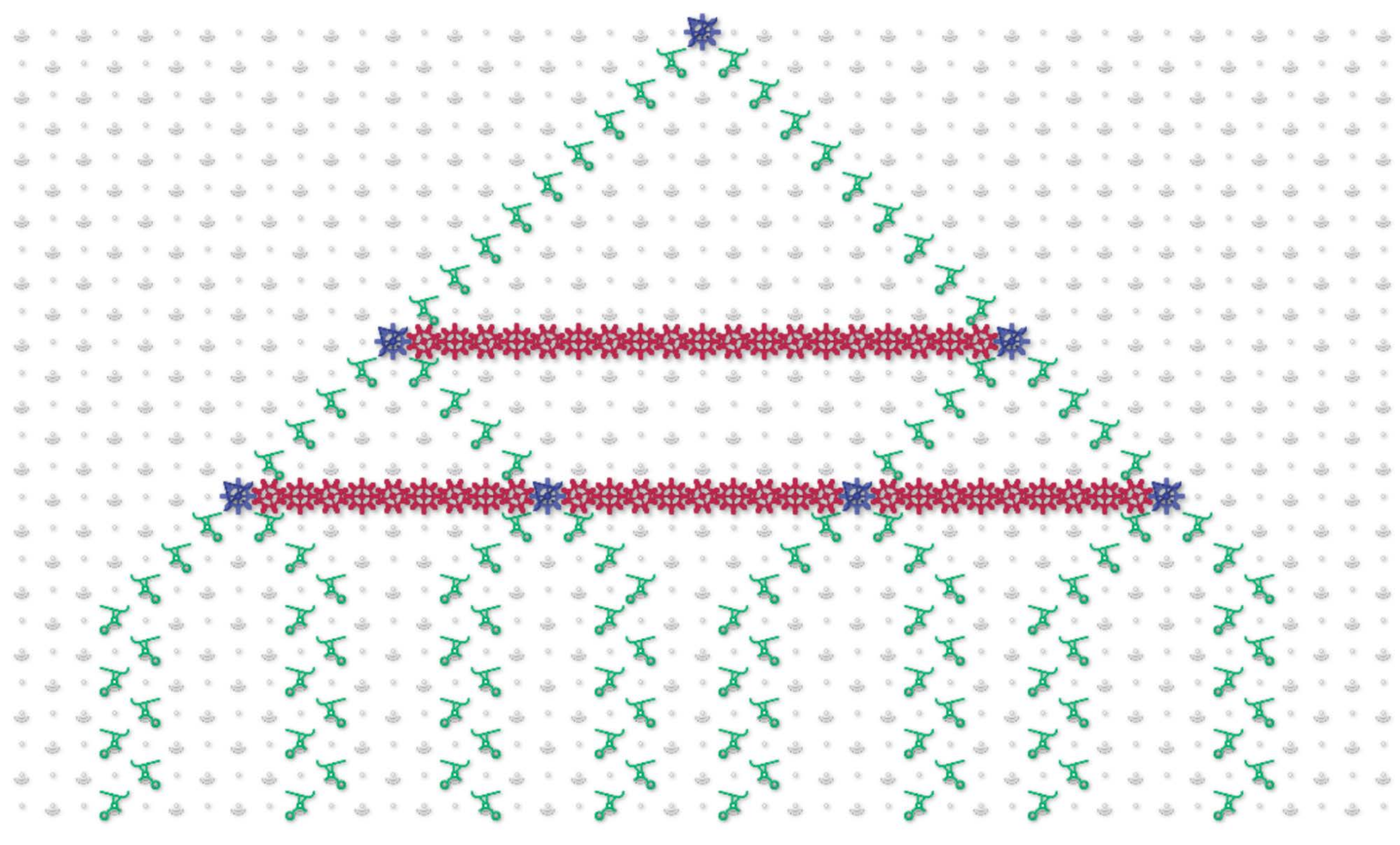}
\end{center}
\caption{A 3-bit branching tree}
 \label{fig:branching-tree}
\end{figure}

Following this reading-and-branching, further along each branch there
can be }{write instructions}{, which are gear bit constructs such as
Write-0 or Write-1, that change the values of the variables already
read, by using a chain of gears stretching back upward to reach the
variable that was read above.  Care must be taken to ensure that this
is done without any chain of gears crossing another, as there is no
mechanism in TT to allow for gear chains to cross without affecting each
other.  Due to the need to avoid crossing, if variables are read and
branched upon in order $x_1, x_2, \ldots{}, x_n,$ then one must write them in
the reverse order $x_n, \ldots{}, x_1$.  This allows for ``nested'' chains
of gears, as seen in Figure~\ref{fig:read-branch-write} below.

\begin{figure}[htp]
\begin{center}
\includegraphics[width=4in]{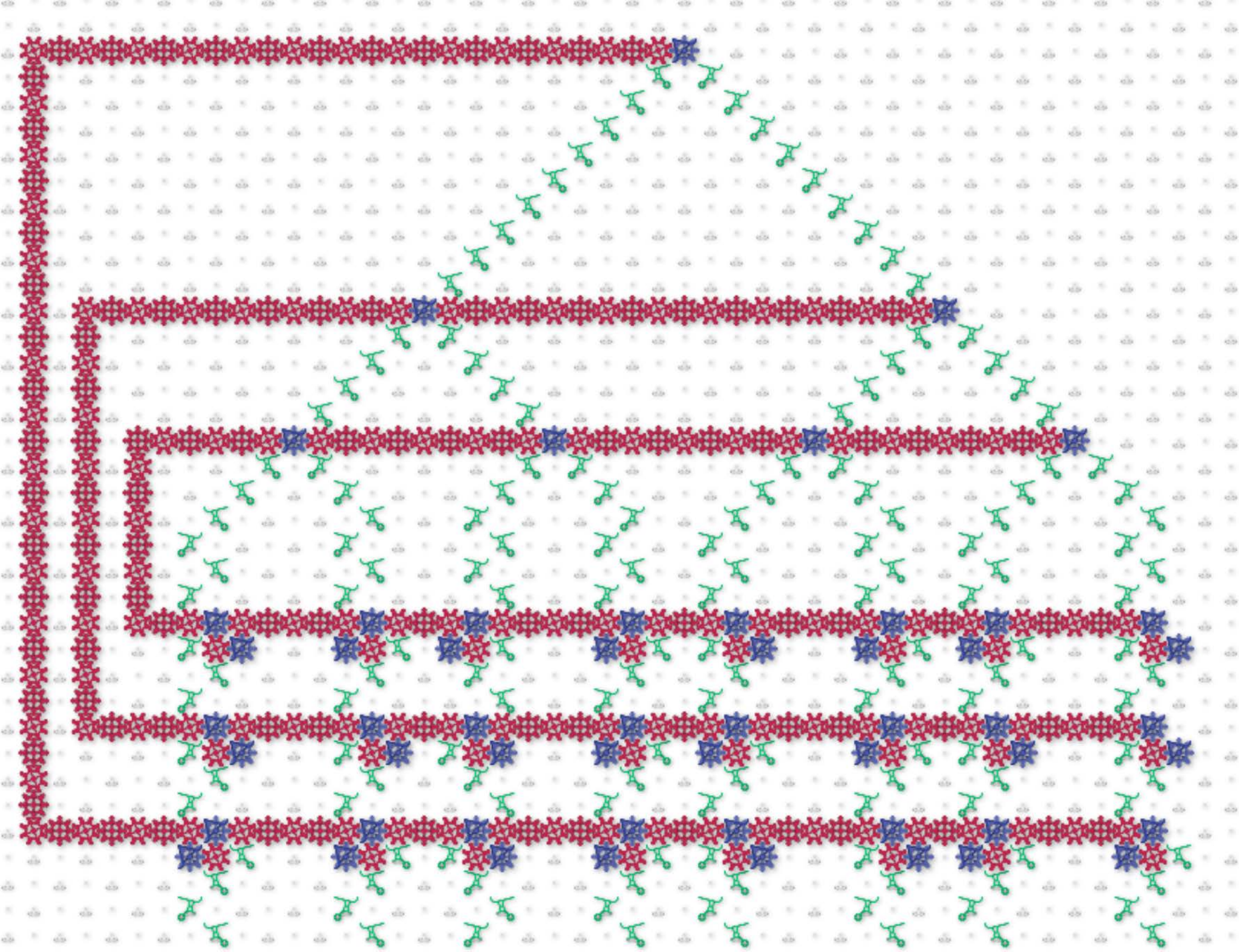}
\end{center}
\caption{Read-Branch-Write with three variables, and three nested
write gear chains.}
 \label{fig:read-branch-write}
\end{figure}

The write-bits of the binary tree are written }{in reverse
order,}{~so that after exiting the tree at the $n$th bit, the first write
will affect the $n$th level of the tree, the next will affect the $(n-1)$st
level, and so on, so that the last write encountered will affect the
very top gear bit in the tree. This forms a collection of nested ``C''
shapes as shown in Figure~\ref{fig:read-branch-write}. 

In Figure~\ref{fig:read-branch-write}, $f(000) = 100$, since if the gear bits in the branching tree are all set to 0 (as they in fact are in the figure), the rightmost
branch is chosen, and then the ball passes through a write-0 to set the
last bit to 0, then another write-0 to set the second bit to 0, and then
a write-1 to set the first bit to 1. 

This ``Read-Branch-Write'' technique is quite powerful, and can easily
be used to show that TT can simulate any finite computer, as follows. 
(The construction and argument in this section originally appeared in a
{\href{https://community.turingtumble.com/t/how-to-simulate-any-computer-with-finitely-many-pieces/429}{post}}
by L. Pitt (July 2018) on the TT Bboard.) 
This construction is
extremely inefficient. Finite computation can be achieved from the much
more efficient reduction of Johnson~\cite{johnson2019}, showing that TT is
P-complete via reduction from the Circuit Value Problem. The
inefficient approach is presented here because the basic ideas of the
read-branch-write loop are used in the next section.

Any computer can be represented via a finite number of bits, giving the
contents of the memory, the registers, the instruction counter, etc.
Suppose there are $n$ such bits total, and that they are stacked into a
single register on the TT board. (So, for example, if the computer had
eight 16-bit registers, we'd have a stack of 8*16 bits at the top
representing their values. Below that might be the finite memory
contents (the program to be executed), the program location counter,
etc. All one tall stack of bits on the board.)

Then after execution of one ``step'' of a program, or one cycle of the
machine, depending on the level of resolution of our view of the
algorithmic process, the bits arrive at a different state. Hence, a
single step of computation on such a computer can be viewed as a
function $f$ from $n$ bits to $n$ bits, and an entire computation is the
result of applying that function repeatedly (i.e., composing the
function with itself sufficiently many times to achieve a fixed point --
until the program has halted and no more changes are occurring).

We can show that for any number of bits $n$, and any computation function
$f$ from $n$ bits to $n$ bits, a finite number of pieces of TT can be arranged
so that a single ball drop will cause the change of a collection of a
stack of $n$ gear bits to reflect the application of function $f$. By
triggering more balls, the device can apply $f$ repeatedly to walk through
the state changes describing the trajectory of the computation.

Consider a computation to be done on $x$, a sequence of $n$ bits
representing the initial state of a computer. As above, we create a
tree of $n$ levels.  Consider }{some particular setting of those bits
corresponding to a snapshot of a computation. Then a ball routed to the
top of the tree will be directed along one of the $2^n$ distinct branches.
From there each branch will continue straight down, and a ball so routed
will pass through $n$ write~components, each of which affects a
possible change in the gear bits in the binary tree as follows: By
choosing either to place a ``write-0'' or a ``write-1'' component in the
branch, the ball can set the bits across a given level of the tree to a 0
or a 1, respectively, via a connected chain of gears. 

Because any sequence we'd like made of write-0 and write-1 components
may be placed along each of the $2^n$ paths, any computation function $f$
can be specified.  And, after one pass of the ball through the tree with
bit pattern $x$ and subsequent exit path corresponding to pattern $x$, and
write bits programmed according to the desired value $f(x)$, by the time
the ball hits the lever below to trigger another ball, the tree pattern
will show $f(x)$ - i.e., one step of computation will have been executed.
Thus, the second ball will apply a second step of computation (computing
$f(f(x))$.  To stop the computation, particular output paths can
terminate with an interceptor instead of continuing on to trigger
another step of computation.

This shows that any computation carried out by a finite computer can be
carried out by Turing Tumble.  It is worth noting that on TTC there
are many (often complex) constructions aimed at simulating a computer in
manners structurally faithful to the components of a computer (i.e.,
registers, program counter, etc.).  The point of our construction was
not to simulate a computer structurally, but rather to demonstrate
easily in a way that can be verified just as easily, that in principle,
TT is powerful enough to carry out any computation parameterized by
bounded input size.

\section{Simulating a Turing Machine}
\label{sec:simulation}

In the previous section we showed how to simulate a finite computer.
~Here we consider the more general question of whether TT is
Turing-Complete, i.e., can any effective computation (as envisioned in
the {\href{https://en.wikipedia.org/wiki/Church-Turing_thesis}
{Church-Turing Thesis}}) be carried out given sufficient space and pieces.  To
establish this, we show how TT can simulate a Turing machine.

\subsection*{Assumption}

For simplicity, we assume a Turing Machine with states $Q_1, \ldots,
Q_n$, with tape symbols only 0 and 1 (i.e., no blank symbol), and one with
a transition function requiring motion to the left or right at each step
(the TM cannot remain on the same cell).  State $Q_1$ is the initial state,
and state $Q_n$ is the unique halt state from which no transitions are
possible.  These assumptions are without loss of generality because any
TM not meeting these requirements can be simulated by one which does,
using simple and well-known techniques.

\subsection*{Running Example}

As we describe the construction, figures will illustrate the technique
on a simple TM with three states that computes the constant 0 function
on unary input.  More specifically, the TM takes a unary input of the
form $01^k$ (the leading 0 is an end-of-tape marker), initially scanning the second tape cell (containing the
first 1), and replaces the 1s with 0s, returning back to scan the second
cell.  It thus computes the constant 0 function: $f(k) = 0$ for all $k$. 
 The TM has the following definition:

\begin{itemize}
\item {States: \{$Q_1, Q_2, Q_3$\}. }
\item {Initial state $Q_1$. }
\item {Halt state $Q_3$. }
\item {Tape alphabet \{0,1\}. }
\item {Transition function $d()$ given by:}
\begin{eqnarray*}
d(Q_1, 1) & = & (Q_1, 1, R)  {\rm ~~Move~right~past~each~1} \\
d(Q_1, 0) & = & (Q_2, 0, L)  {\rm ~~ Read~ past~ 1s,~then~go~left~and~change~to~Q_2}\\
d(Q_2,1) & = & (Q_2, 0, L) {\rm ~~ Move~left~past~1s,~changing~them~to~0s}\\
d(Q_2,0) & = & (Q_3, 0, R) {\rm ~~ Left~ edge~marker~found,~move~right~and~halt}\\
d(Q_3,*) & = & < null > {\rm ~~ Halting~state;~no~transitions~applicable}\\
\end{eqnarray*}
\end{itemize}

\subsection*{The TM Tape}

We simulate the tape cells of the TM with an infinite sequence of
``3-cells'' going down the board. Each 3-cell carries three bits of
information: 

\begin{itemize}
\item
  {Whether or not the head is at this 3-cell (bit ``$h$''). Except when
  it is changing, exactly one 3-cell will have bit $h$ set to 1 at all
  times.}
\item
  {What the symbol is at this 3-cell (bit ``$S$'')}
\item
  {Whether the head is at the previous 3-cell (bit ``$p$''). 
  This bit is connected via a gear chain to the previous 3-cell's copy
  of $h$, and is used to control it.  The first 3-cell does not have
  this bit, as there is no previous 3-cell.}
\end{itemize}

We use three ball colors, which we call ``0'', ``1'', and ``Query''. 
In each 3-cell, there is a trigger for each one of these ball hoppers,
called ``T0'', ``T1'', and ``TQ'', respectively.

\subsection*{The Finite Control}

If  TQ is triggered, the finite control, situated at the top of the
board, must query the tape to determine what the current symbol is.
This is achieved by dropping a Q-ball down a ramp chain, going by each
3-cell, until one with $h=1$ is found, indicating the TM head is at this
position. The bit representing the symbol $S$ at this 3-cell is read, and
depending on whether its value is 0 or 1, a ball is dropped via trigger
T0 or T1, respectively, back into the finite control. Then the Q-ball
is intercepted.  Every other ball dropped will be a Q-ball, with the
main purpose of determining the current symbol read, and releasing the
corresponding 0 or 1 ball. We will see exactly how this is
implemented later.

The finite control has two main components: A State Transition Module
(STM), and a Router.  The STM keeps track of the state of the TM, via $n$
rows of gear bits, with each row connected across with gears.  The $k$th
row represents the state $Q_k$, and exactly one of these rows will be set
to 1 at any time (except perhaps when changing). As in Section~\ref{sec:finite-computer}, the
rows are in order $Q_1, Q_2, \ldots{}, Q_n$, going down the board and are
organized into a read-branch-write tree, using the ``write'' gear bits
~$Q_n', \ldots{}, Q_1'$,  (in that order), connected by C-shaped gear chains so that
changing $Q_i'$ will change $Q_i$, as in Figure~\ref{fig:read-branch-write}. 
However, because only
one of the $Q_i$ will be set to 1 at any time, there need be only $n$ output
branches, instead of $2^n$

We actually use}{~}{two}{~such trees, one below the 0-ball release, and
one below the 1-ball release.  The rows for each $Q_i$ span across both
trees, so all copies of $Q_i$ are connected via gears in any given row. 
~In order to connect the two trees, the nested  C gear chains of the
read-branch-write trees must not interfere, so the tree on the left has
the nested gear chains on the left, and the tree on the right has the
nested gear chains on the right (with the shape of a reversed C).  The end result is that each variable
is in a cycle of gear bits, as in Figure~\ref{fig:STM} below.

\subsection*{Changing the State }

Suppose a $b$-ball ($b$ = 0 or 1) encounters some $Q_i=1$ and follows the
appropriate branch.  Further suppose that $d(Q_i, b) = (Q_j, b', D)$ where
$D$ = L or R representing move-left or move-write.  Then as the ball
follows the branch past $Q_n'$  to $Q_1'$, it sets $Q_j'$ to 1, and all other $Q_i'$
to 0.   Due to the gear chains, this will set $Q_j$ to 1 as the machine
gets ready for the next transition.  Figure~\ref{fig:STM} shows how the states are
updated in the first STM in the simulation of the 3-state TM
defined above.

For the moment, ignore the tape symbol written and the direction of
motion of the head, and focus only on the state transitions. For
example, $d(Q_1,0) = (Q_2, \ldots)$ indicating that if the machine
were in state $Q_1$ and reading a 0, then the next state would be $Q_2$. 
This transition is circled in the figure. 
\begin{figure}[htp]
\begin{center}
\includegraphics[width=4in]{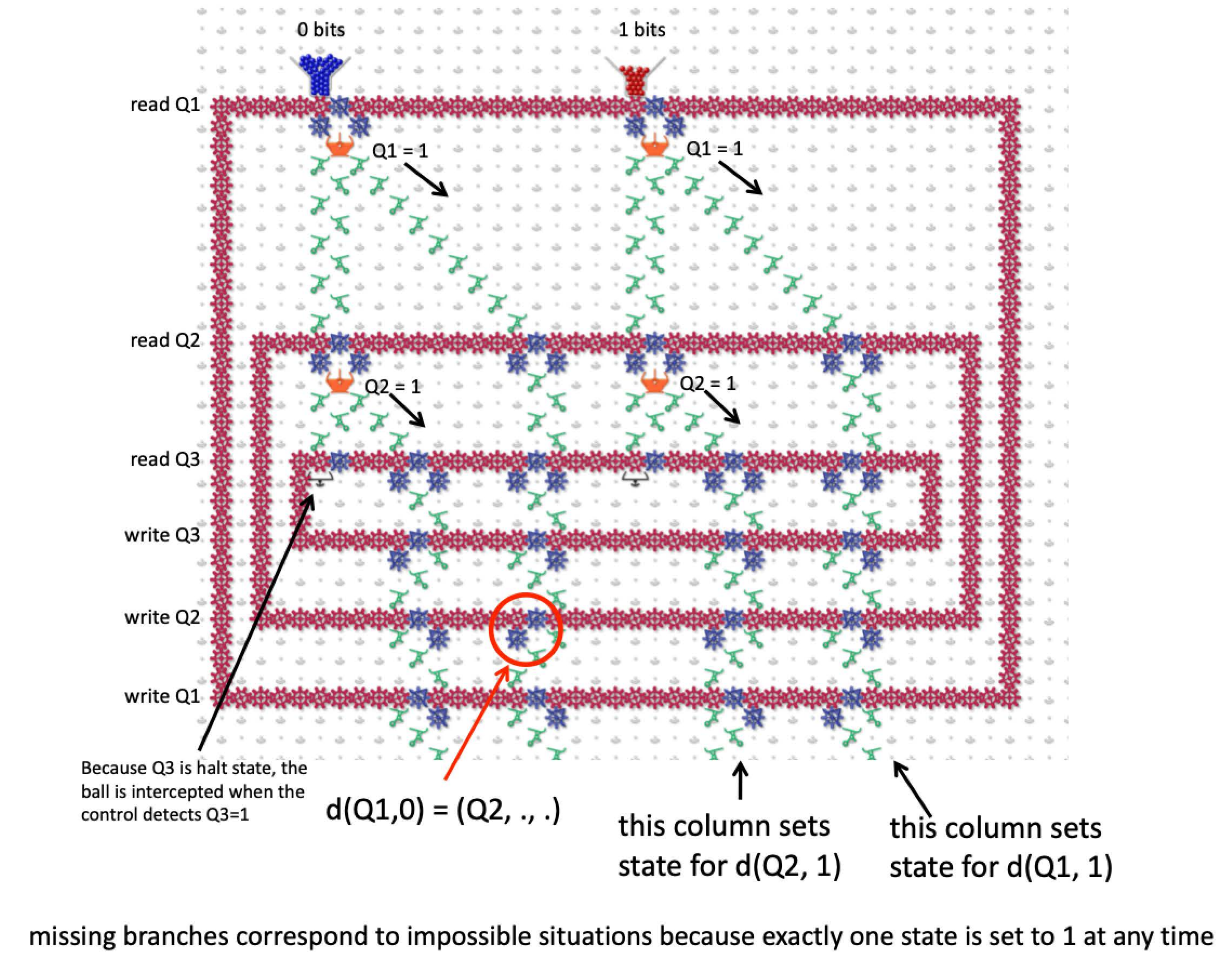}
\end{center}
\caption{State Transition Module (STM) of finite control for a 3-state machine.   The $2(n-1)$ (in this case 4) columns now pass down to router, which sends the ball on one of 4 lines:  L0, L1, R0, R1, indicating whether the TM is moving L or R, and whether to update the scanned symbol to 0 or 1.
}
 \label{fig:STM}
\end{figure}

The STM sets the state for the next step, but we have yet to execute
the current step. Note that the column the ball is in within the STM
implicitly stores that the machine used to be in state $Q_i$, even though
the state variable was updated to show $Q_j$.  Between these two halves
then, we now have a ball emerging from one of $2n$ columns, carrying
information of the form (current state, symbol scanned).  (Actually,
$2n-2$ columns. If $Q_n$ is reached, there is no branch, but there is an
interceptor immediately below it, since it is the halt state.)

For example, if the current state (column) were $Q_i$, and the ball was
in the right half of the STM (indicating a 1-ball had been released),
then we know that the TM is in state $Q_i$ reading the symbol 1.  
Suppose the transition function specifies $d(Q_i, 1) = (Q_j, 0, L)$. The
ball traveling through the state change component has already set the
next state to $Q_j$. How can it indicate that the TM should write the
symbol 0 to the current cell and then move left? We do this by using
the position of the ball on the board. 

There are only four possibilities for what the TM could do at this
point: It could write 0 or 1, and it could move left or right.  Denote
these actions as L0, L1, R0, and R1, with the obvious correspondence.

Using ramps and crossovers in a ``Router'', we can route a ball
emerging from any column of the read-branch-write tree onto any specific
one of four vertical paths we use to represent these four possibilities,
and the ball will travel down the path designated by the transition
function given its current state and symbol read.  In this way the
3-cells below can determine the actions to be taken, based on which of
the paths L0, L1, R0, R1 the ball arrives on.  In other words, a ball
traveling down one of these paths carries the following information to
the 3-cell with $h = 1$, and means the following should occur:

\begin{itemize}
\item  L0: Means that the symbol $S$ being scanned should be set to 0, and
  the TM should move its head to the left
\item  L1: Similar, but set $S$ to 1, and move left
\item  R0: Set $S$ to 0, and move right
\item  R1: Set $S$ to 1, and move right.
\end{itemize}

For example, if the 1-ball goes through its tree, and encounters $Q_2 =
1$, and if $d(Q_2, 1) = (Q_7, 0, R)$, then the ball is routed so that (a)
$Q_7'$ is set to 1, and then (b) the ball is then routed on to the output
line R0.

Figure~\ref{fig:router} below shows the Router portion of the 3-state TM whose STM
was shown in Figure~\ref{fig:STM}.  As a reminder, the transition function $d$ has
definition::\\
\hspace*{.5in}$d(Q_1, 1) = (Q_1, 1, R)$\\
\hspace*{.5in}$d(Q_1, 0) = (Q_2, 0, L)$\\
\hspace*{.5in}$d(Q_2, 1) = (Q_2, 0, L)$\\
\hspace*{.5in}$d(Q_2, 0) = (Q_3, 0, R)$\\
\hspace*{.5in}$d(Q_3, ?) = <null> $

\begin{figure}[htp]
\begin{center}
{\includegraphics[scale=0.15]{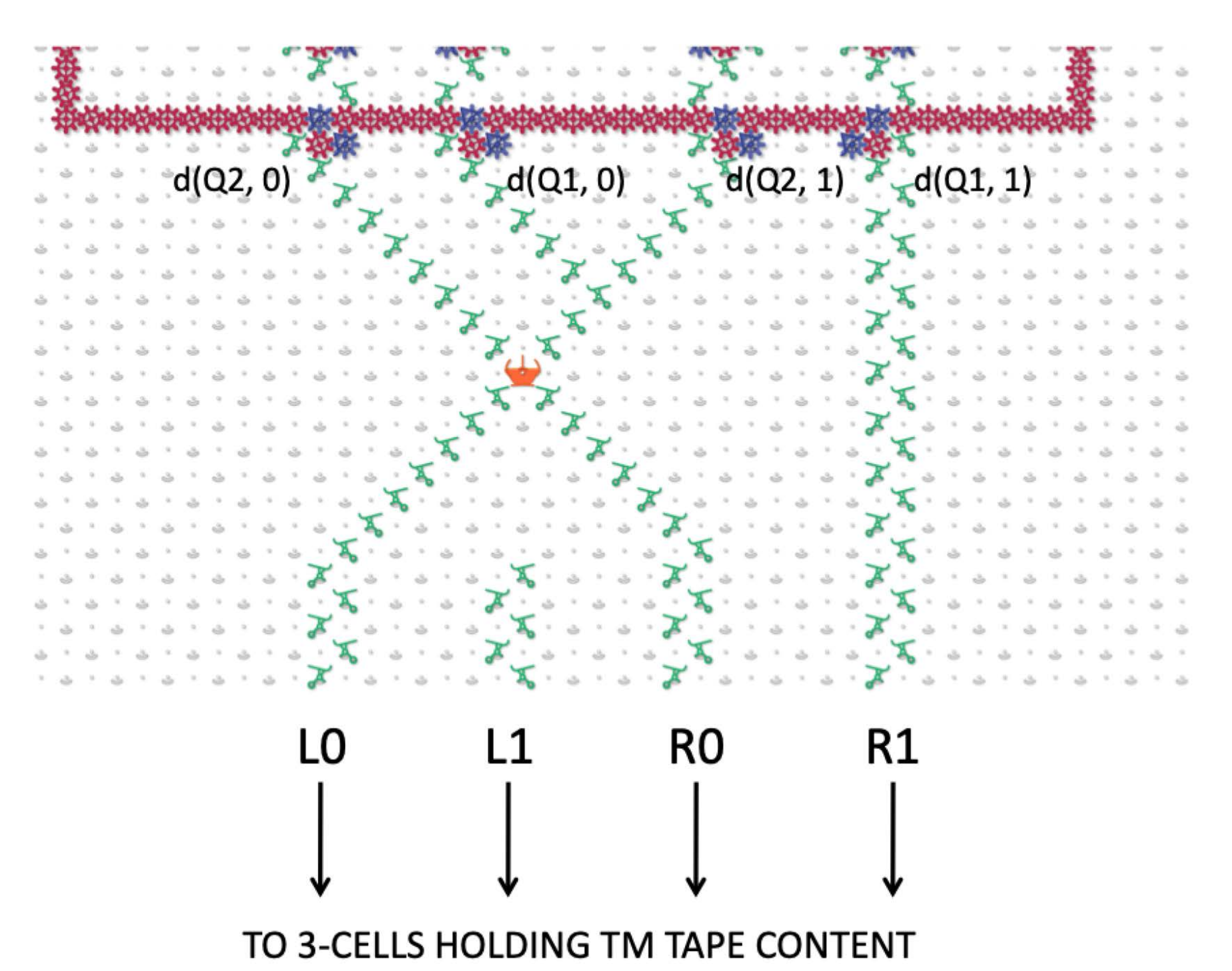}}
\end{center}
\caption{The Router portion of the finite control of a 3-state TM.   No transition dictates that the TM moves
left and outputs 1.}
 \label{fig:router}
\end{figure}

The paths L0, L1, R0, R1, (and Q - not shown in Figure~\ref{fig:router}) will collectively form
the ``Bus'' which runs down the left side of the tape structure forever,
going by bit gears representing the bit indicators $h$ of the different
3-cells, until one which is set to 1 is encountered, indicating the ball
has arrived at the 3-cell representing the scanned symbol on the TM's
tape. Refer to Figure~\ref{fig:bus} below.

\subsection*{Updating the Tape}

We now show how the 3-cell ``tape'' is updated.  At this point, the TM
has transitioned to the correct new state, and a ball is falling down
one of five lines of the Bus: L0, L1, R0, R1, or Q, indicating either
that the finite control is asking for the value of the current cell
(ball on Q line), or delivering the information about the new symbol to
be written and the direction the head should move (ball on L0, L1, R0,
or R1).

As noted above, each 3-cell carries 3 bits of information: $h, S$, and
$p$, representing whether the TM's tape head is at this cell ($h = 1$) or
not ($h = 0$), what the symbol is at this cell ($S = 0$ or 1), and the
value of the variable $h$ in the previous cell ($p = 0$ or 1).  The first
3-cell has a $p$-variable, but it is not attached to anything.  Thus, we
have three rows of gear bits, one for each of $h, S,$ and $p$, in that
order, going down.

The row of gear bits for $h$ cross the Bus.  At the intersection of L0,
L1, R0, or R1 with $h$, if $h$ = 1 then a Write-0 instruction is
executed and the
ball leaves the Bus and heads down to the right and into the 3-cell.
The line Q on the other hand does a{\em  non-destructive} read of $h$, and branches right if $h =
1$.  If $h=0$ the ball continues straight down the Bus towards the next
3-cell.   These right branches create five parallel lines (called L0$'$,
L1$'$, R0$'$, R1$'$, and Q$'$) interleaved with the original ones.  We continue
these new branches diagonally down to the right, so that they cross over
all of the original lines via crossovers, heading toward the 3-cell
containing this $h$. (Figure~\ref{fig:bus}).

\begin{figure}[htp]
\begin{center}
\includegraphics[scale=0.3]{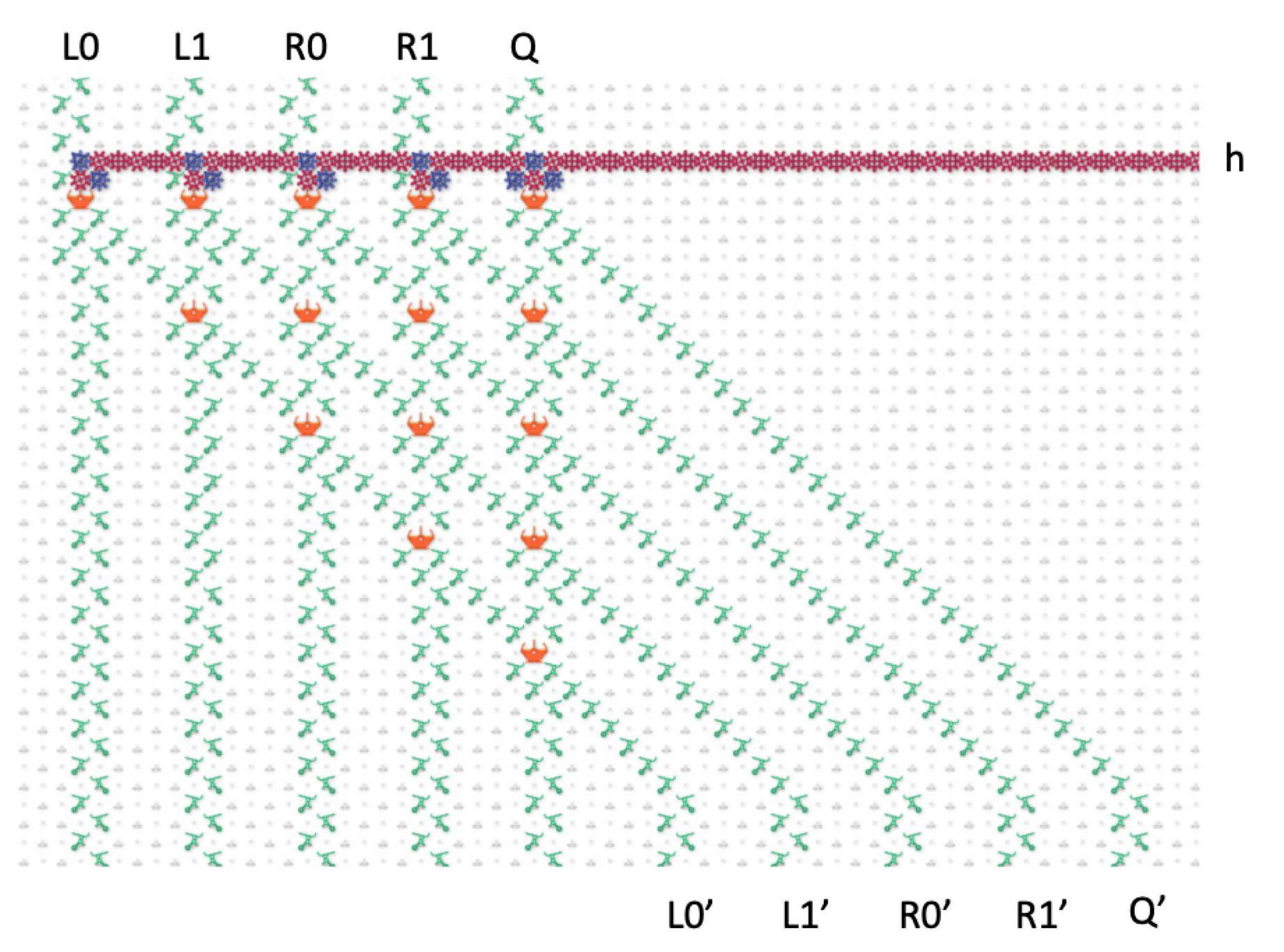}
\end{center}
\caption{Five instruction lines entering a 3-cell because $h=1$.}
 \label{fig:bus}
\end{figure}

These five new lines are used to execute the transition, by writing the
new symbol, and changing the head position.   The original lines
follow the $h=0$ branch down and over any rows of gear bits they may
encounter by passing through the ``ignore'' construct, and then these
lines go on to interact with the $h$ variable for the next 3-cell.

We describe what happens with the new branches 
L0$'$, L1$'$, R0$'$, R1$'$, and Q$'$,
 with the first four indicating the symbol to be written, the
direction to be moved, and the knowledge that $h=1$ (but has
now been set to 0 by the Write-0 instruction) so this is the 3-cell
representing the TM tape cell that is being scanned.  Note that in
these four cases, at the next time step, $h=0$ is correct, since the TM
must move left or right, so will no longer be at the current cell. 

Referring to Figure~\ref{fig:3cell}, the next row of bits encountered is for $S$.  
~If the ball arrived on line Q, then the finite control is querying what
symbol is at this cell. In this case, the ball does a read of $S$,
~branching left ($S=0$) or right ($S=1$), triggers T0 or T1 respectively,
and then is intercepted. Thus, the next ball to fall indicates the
symbol just read. 

\begin{figure}[htp]
\begin{center}
\includegraphics[scale=0.15]{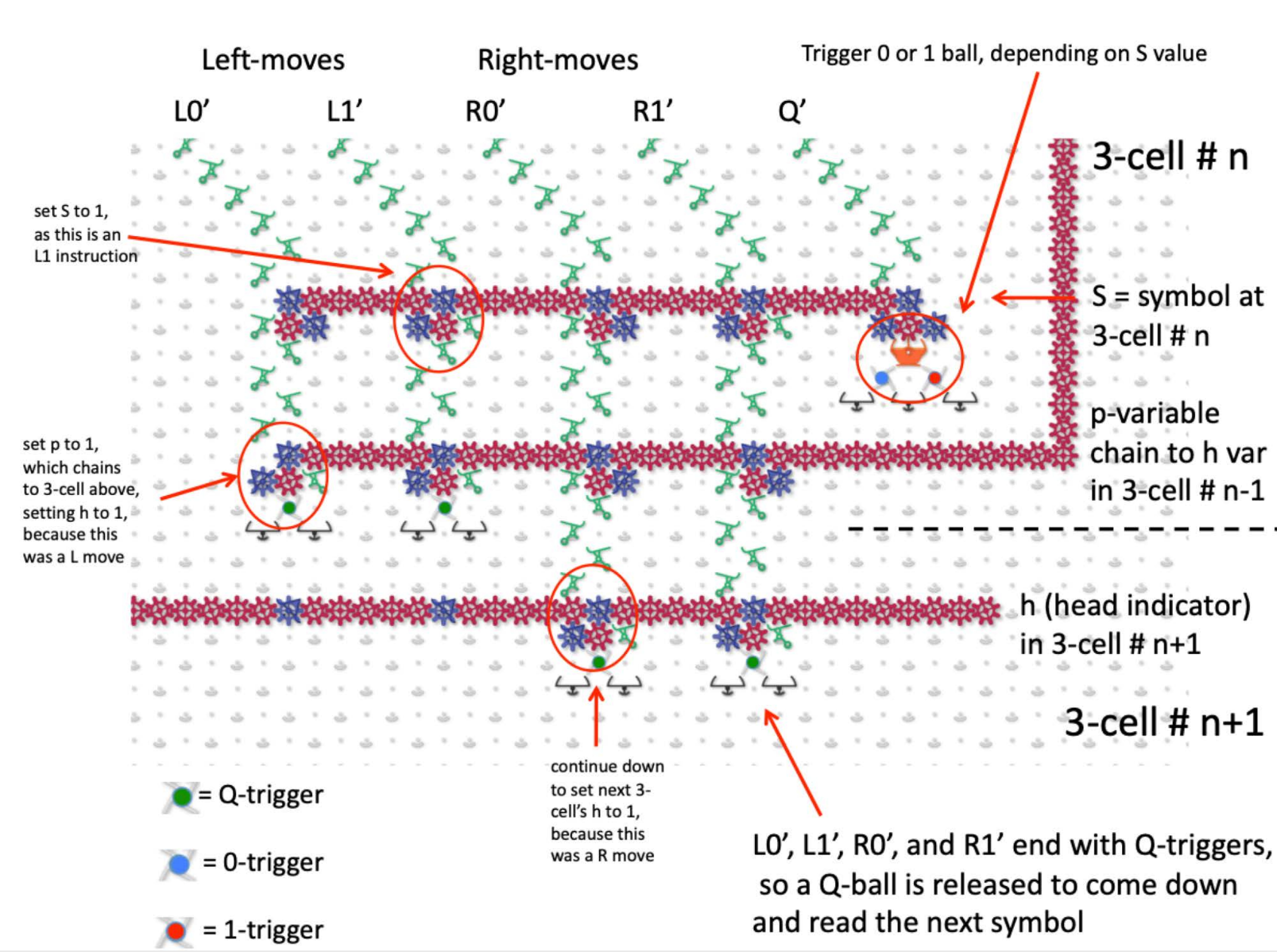}
\end{center}
\caption{Inside a 3-cell. }
 \label{fig:3cell}
 \end{figure}
 
If instead the ball arrived on one of the other four lines, the update
instructions are carried out in the following way. L0$'$ and R0$'$ each
write-0 as they cross $S$.  L1$'$ and R1$'$ each write-1 as they cross $S$.

At this point, with the ball just below the $S$ row, we have that the
symbol $S$ has been updated, $h$ has been set to 0, and all that remains to
be done is to update the position of the head, and trigger a new query
ball.  The next row of gear bits is $p$, which is also connected to the
}{previous}{~3-cell's $h$.  Thus L0$'$ and L1$'$ write-1 as they cross over
$p$, thereby setting the previous 3-cell's copy of $h$ to 1.  They then
pass through a TQ trigger, and are intercepted. This completes the step of
the simulation for these two lines.

On the other hand, a ball on the R0$'$ or R1$'$ lines crosses over $p$,
writing 0 (as the previous cell will not be the next location of the
tape head), continues down to the next 3-cell, encounters the first row
of gear bits (for $h$), writes 1, then triggers TQ and is intercepted. 
This indicates that the head has moved right (down) one cell on the TM
tape, and the corresponding 3-cell now has its $h$-bit indicator set to
1.

That completes the construction, but we must be attentive to\ldots

\subsection*{Topology}

We need to show that the above structure actually can be laid out. 
The figures show much of the layout.  Ramps can cross each other with
crossover pieces (as seen), and ramps can cross gear chains using Ignore
constructs.  However, there is no mechanism for different gear chains
to cross without interfering with their function, so the main possible
concern is that the connection of the $p$ gear bit chains to the previous
$h$ gear bit chains will somehow intersect some other gear chain(s), or
each other.  This can be avoided by sending these back-going gear
chains up alternate sides of the 3-cells and Bus as we now describe.

If the 3-cells are numbered, then extend the $h$-chain of odd numbered
3-cells, and the $p$-chain of even numbered 3-cells, far enough out to the
right to have a clear vertical path between them, and connect the
$p$-chain of the $n$th 3-cell with the $h$-chain of the $n-1$st 3-cell, to form
a ``$p$-$h$'' gear chain.  (The first 3-cell does not have its $p$-chain
connected to anything.)  Also extend to the left (all the way past the
Bus) the $h$-chains of even numbered 3-cells, and the $p$-chains of odd
numbered 3-cells, also connecting the $p$-chains in the $n$th 3-cell to the
$h$-chains in the previous $n-1$st 3-cell, and also forming a $p$-$h$ gear
chain.  Because one is on the left, the other on the right, this
ensures that the two adjacent $p$-$h$ gear chains do not have to cross. 
Note that the $S$-gear chains are short and isolated, and they don't
extend in either direction. Refer to Figure~\ref{fig:topology}.

\begin{figure}[htp]
\begin{center}
\includegraphics[scale=0.3]{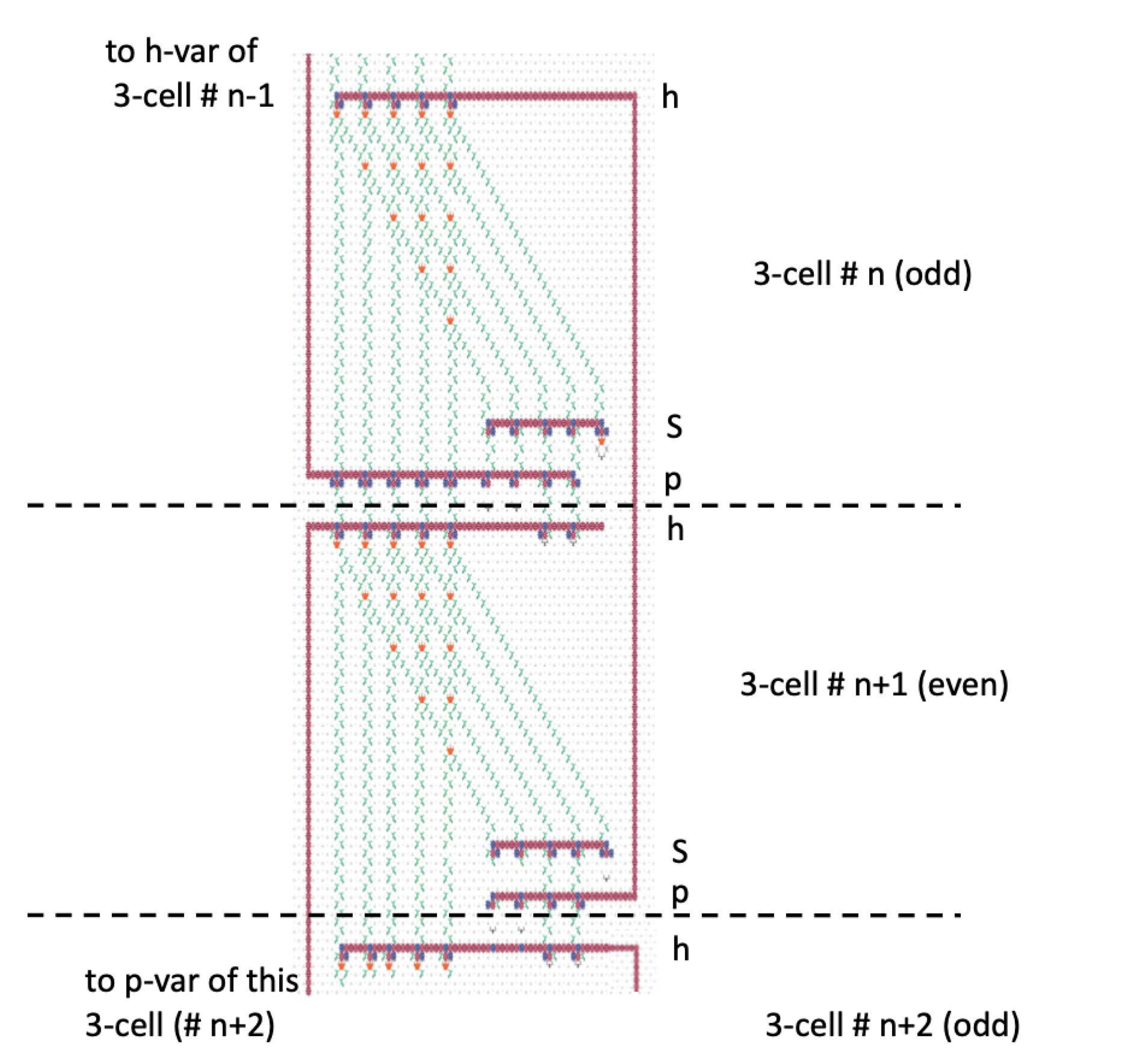}
\end{center}
\caption{~$p$-$h$ gear chains between adjacent 3-cells.}
 \label{fig:topology}
\end{figure}

\subsection*{Putting it All Together}

Figure~\ref{fig:architecture} shows the top part of the entire machine for the 3-state TM
example, with the first two 3-cells included.  A TM with a much larger
number $n$ of states would not look too different; in fact, the only
difference is that the STM would have $n$ loops with $2n-2$ branches
entering the router, which would still route down to the four output
branches L0, L1, R0, and R1.  The 3-cells would be identical.
\begin{figure}[htp]
\begin{center}
\includegraphics[scale=0.45]{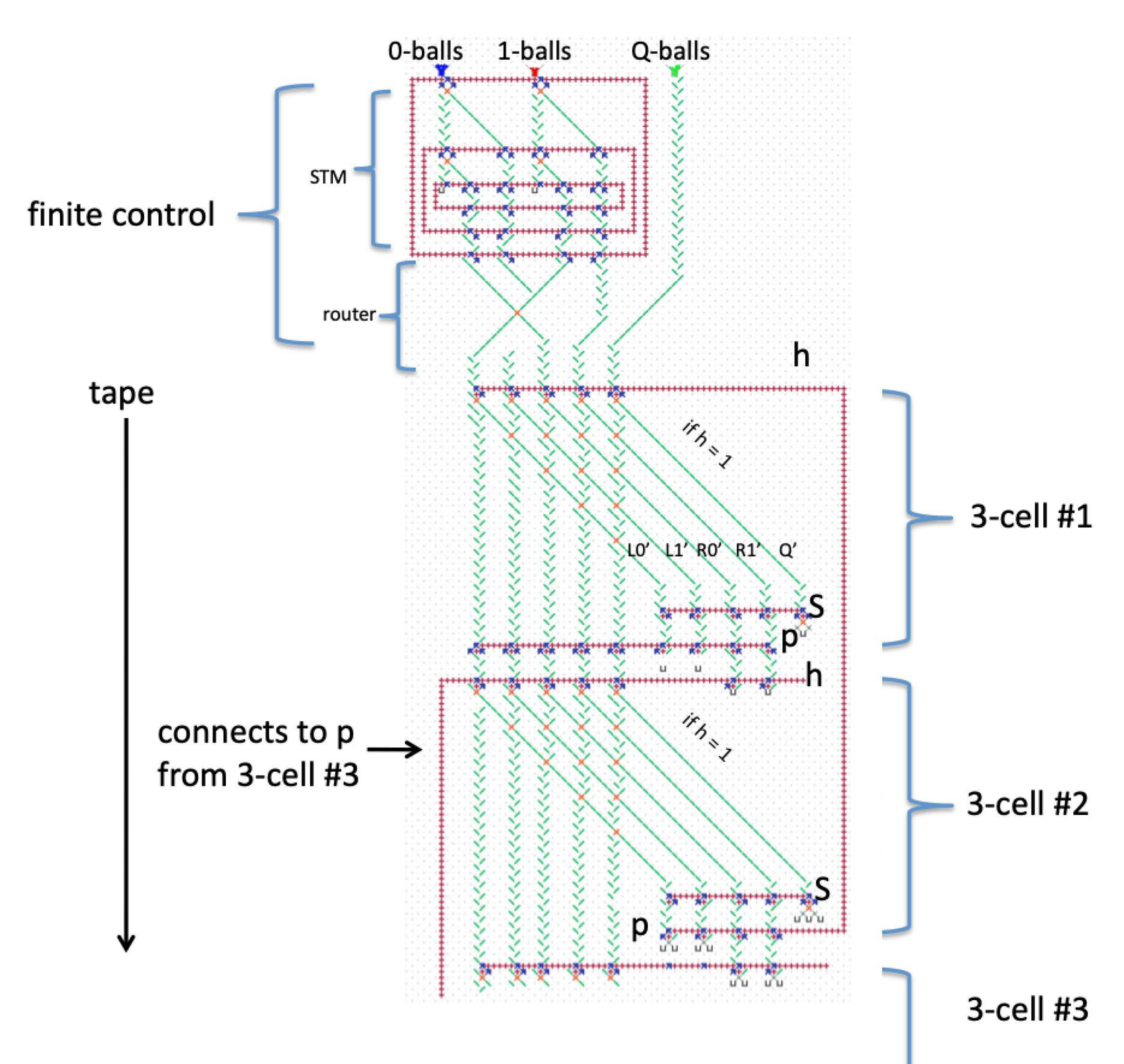}
\end{center}
\caption{Architecture of Turing Machine Simulator using Turing
Tumble.}
 \label{fig:architecture}
\end{figure}

\subsection*{Initial Set-up, and Getting Started}

Initially, we set $Q_1$ to 1, and all other $Q_i$ to 0 in the finite control.
 If the input word has $k$ symbols, we set the occurrences of $S$ in the
first $k$ 3-cells to reflect that word, and all other $S$ values to 0. 
 We set $h$ in the first 3-cell to 1, all other $h$ values to 0.  For all
cells, the value of $p$ is geared to the previous $h$, so they will be
initialized as the $h$'s are..namely $p=1$ in the second 3-cell, and $p$ = 0
in all others. }    {Thus, beyond the $k$th 3-cell, all gear bits are initialized to 0.

To start the simulation, we press the TQ trigger, which will release a
ball down the Q line, and cause the machine to query the first 3-cell's
symbol, and we're up and running.

\subsection*{Halting}

In the STM, if $Q_n$ is ever reached, it encounters an interceptor.

\section{Why Does This Work?}

{We will show that the simulation is correct. We need several
supporting Lemmas.}

\begin{lemma}
\label{lem:oneQ}
Except while a ball is traveling within the STM, there is
exactly one state variable $Q_i$ that has the value 1.
\end{lemma}


\begin{proof}
Initially, there is one $Q_i$ set to 1.  The only place that
the $Q_i$ are set is in the STM. This occurs as a ball travels down a
column during the write phase of the read-branch-write structure. But
each column in that structure sets all $Q_i$ to 0 except for exactly one. \qed
\end{proof}

\begin{lemma}
\label{lem:oneh}
After each ball is intercepted, there is exactly one $h$ set to
1.
\end{lemma}

\begin{proof}  At the beginning of the simulation, there is exactly one
value of $h$ set to 1 (in the first 3-cell).  Each Q-ball doesn't change
any $h$, so this property remains true when Q-balls are intercepted. 
~Where is $h$ set to 0? If $h=1$, it is set to 0 while the ball branches
onto one of L0$'$, L1$'$, R0$'$, R1$'$.  Suppose the ball is on L0$'$ or L1$'$.
These write 1 to $p$, thus changing the previous $h$ to 1, and the ball is
then immediately intercepted.  So, after such a ball is intercepted,
there is still exactly one $h$ set to 1.  (The only problem is that if
the ball is in the first 3-cell, and along the line L0$'$ or L1$'$, the
instruction is to go left. But the first 3-cell's copy of $p$ does not connect to any $h$.  But, it shouldn't, because the TM is at the leftmost cell.
Thus, the TM is experiencing a crashing computation, trying to move
left off of the edge of the tape, and the computation halts by
convention.  There should be no next head position.)

On the other hand, if the ball is on R0$'$ or R1$'$, then the $h$ variable in
the next 3-cell is set to 1, and the ball is then intercepted (after a
Q-trigger). Thus, again there is exactly one $h$ value set to 1.
\qed
\end{proof}

\begin{lemma}
\label{lem:alternating}
Every odd ball released is a Q-ball. Every even ball released is either a 0 or 1 ball. 
\end{lemma}

\begin{proof} The proof is by induction. The first ball is a Q-ball,
because the simulation starts with the press of the Q-trigger.  Suppose
that for all even and odd values $k$ up to and including $k = 2n-1$, the
Lemma is true.  We will show it remains true for $k = 2n$ and $k=2n+1.$
Since the Lemma is true up through $k = 2n-1$, we know that the $2n-1$st ball released is a Q-ball.  This Q-ball will find the unique $h=1$ (Lemma~\ref{lem:oneh}), branch, and trigger T0 or T1 after
reading $S$. Thus, the $2n$th ball is a 0-ball or a 1-ball.

Now consider this 0-ball or 1-ball as it goes through the
finite control.  Exactly one $Q_i$ is set to 1 (Lemma~\ref{lem:oneQ}). 
If this is $Q_n$, then the simulation halts so the Lemma is true vacuously, there being no additional even- or odd-balls.  If it is not $Q_n$, then the ball
must be routed down one of L0, L1, R0, or R1. Exactly one $h$ is set to 1
~(Lemma~\ref{lem:oneh}), so the ball when encountering that 3-cell, will branch onto L0$'$, L1$'$, R0$'$, or  R1$'$.  Each of those lines terminate with a TQ and an interceptor, so the next (the $2n+1$st) ball is a Q-ball.  This completes the inductive step, and the proof.
\qed
\end{proof}

\begin{theorem}
\label{thnm:correct}
For all $n$, after the $2n$th ball is intercepted, the TT
structure accurately reflects $n$ steps of the computation of the TM with
the given input.  More specifically, in the TT structure, 
\begin{enumerate}
\item[(a)]
  {The unique $Q_i$ set to 1 is the state that the TM would be in after $n$
  steps of computation.}
\item[(b)]
  {The 3-cell that contains the unique $h$ set to 1 is the one that
  corresponds to the cell being scanned by the TM after its $n$th step of
  its computation.}{~ }
\item[(c)]
  {The settings of the $S$ variables in all 3-cells accurately reflect the
  contents of the TM tape after $n$ steps of computation.}
\end{enumerate}
\end{theorem}

\begin{proof}
By induction. This is initially true, when $n=0$.  Assume the
Theorem is true for all values of $k$ less than $n$.  Consider $k = n$, and
the $2n$th ball.  By Lemma~\ref{lem:alternating}, the $2n$th ball is a 0-ball or a 1-ball. We
will show that it is the correct ball (0-ball or 1-ball) for where $h$ is
at the $2n-1$st step of the computation. That is, it correctly reflects
the symbol at the currently scanned cell, just before the $2n$th ball is
intercepted.

By our inductive assumption, after the $2n-2$nd step, $n-1$ steps of the
computation have been correctly simulated. Suppose that at that point,
the TM is in state $Q_i$, is scanning the $x$th tape cell, and that cell
contains the symbol $b$ (= 0 or 1).  Thus, the variable $Q_i$ is set to 1,
the value of $h$ at the $x$th 3-cell is set to 1, and the value of $S$ at that
3-cell is $b$.

Suppose that $d(Q_i, b) = (Q_j, b', D)$.

Now the $2n-1$st ball is a Q-ball (by Lemma~\ref{lem:alternating}), which will ``query'' the
tape contents by sending a ball down the line Q.  It will encounter $h=1$
at the $x$th 3-cell, branch off to the right onto line $Q'$, read the symbol
$S=b$,  and roll through the correct trigger $Tb$ releasing a $b$-ball ($b = 0$
or 1)

Thus, the $2n$th ball enters onto the correct side of the STM -- in the
left half if $b = 0$, or the right half if $b = 1$.  Since $Q_i$ = 1 in the
TT structure, and the side of the STM corresponds to the symbol $b = 0$
or 1 being read by the TM after $n-1$ steps of the simulation, the $2n$th
ball will correctly set the state to $Q_j$ as it goes through the STM,
satisfying (a).  Furthermore, the ball will then enter the Router,
which will send the ball down line L0, L1, R0, or R1, by construction,
depending on the values of $D$ and $b'$ (in other words, it will send the
ball down the line $Db'$).

Recall that $h = 1$ in the $x$th 3-cell, representing the $x$th tape cell of
the TM, being scanned after $n-1$ steps of computation ($2n-2$ balls of the
TT).   The $2n-1$st ball, a Q-ball, did not change the $h$ bit. Thus, as
the $2n$th ball falls down the correct one of the L0, L1, R0, or R1 ramps,
~the $h=1$ that it encounters will still be in the $x$th 3-cell.  Note that
$h$ is set to 0 as the ball branches on $h$.

{If the ramp that the $2n$th ball was on was L0 or L1, then by
construction it sets $S$ to 0 or 1 respectively, and then sets $p=1$ in the
$x$th 3-cell. Thus, after the $2n$th ball passes, the tape contents in the
cell are correctly updated. Further, sInce $p$ is connected to $h$ in the
$x-1$st 3-cell, the structure now represents that the TM is scanning the
$x-1$st 3-cell, thus (b) and (c) are satisfied.  }

If on the other hand the $2n$th ball was on R0 or R1, then it too sets $S$
correctly (satisfying (c)), but goes on to the $x+1$st 3-cell and sets $h$
to 1, after which it triggers TQ and is intercepted.  Since the TM
instruction was to move right, the $h=1$ in the $x+1$st 3-cell accurately
reflects that the TM is now scanning cell $x+1$ of its tape. Hence, (b)
is satisfied. 

Since (a), (b), and (c) hold after the $2n$th ball is intercepted, This
completes the induction, and the proof of the theorem.
\qed
\end{proof}

\begin{corollary} 
\label{cor:simulates}
The TT construction computes whatever is computed by the
TM being simulated.
\end{corollary}

\begin{proof}
If on input $w$, the TM runs for $h$ steps, and halts, with
output $w'$ on its tape, then by the Theorem, after $2h$ balls, the TT's STM
records the state is $Q_n$, and the 3-cells represent the string $w'$ via
their $S$ bits..   If the TM when run on $w$ is undefined, because it
either doesn't halt, or it moves off the left edge of the tape, then the
TT simulation either doesn't halt (never reaches state $Q_n$), or else it
sets the $p$ bit in the first 3-cell, which isn't connected to any $h$ bit,
and then the simulation reaches a point at which there is no $h$-bit set
to 1, so the next ball will just keep traveling down forever. 
(Alternatively, we could put an interceptor if an L0 or L1 move were
detected from the first 3-cell, indicating a move-off-left-edge-of-tape
condition, and then the simulation would halt, but we could determine
that the computation were invalid because there would be a ball in that
interceptor.) In any case, the simulation of the TM on this input $w$ is
also undefined.
\end{proof}

\begin{corollary}
\label{cor:complete}
Turing Tumble is Turing-complete.
\end{corollary}

\begin{proof}
Follows immediately from Corollary~\ref{cor:simulates}.
\qed
\end{proof}

\section{Improvements: Using Only One Trigger}
\label{sec:improvements}

In the Introduction, we discussed what a ``reasonable'' extension of
the finite toy would be in the context of allowing unbounded
computation. An infinite number of parts such as bits, for example,
seems reasonable, because otherwise only a finite amount of information
can be stored, and we have nothing more than a finite state machine. 
Triggers somehow seem different. Because they each need to be
physically connected to a ball hopper, placing them arbitrarily far down
seems problematic, as does placing an infinite number of them, as the
connecting pieces behind the board, which we are not modeling, would
interfere with each other. It seems that we need to assume some type of
magic connection, able to function across arbitrary distances, between a
freely-placed trigger and a ball hopper in the case of infinitely many
triggers.

This motivates the question of whether a more reasonable simulation is
possible - can we use only finitely many triggers?   An
 {\href{https://docs.google.com/document/d/1HmGJwkIMCrnXdrB20eZuiuXYMDz4DddypYoSYGbWLiM}
{early version of this paper} 
used finitely many triggers, and
could be improved to using only one trigger, but seemed to lack
reasonability in that the entire finite control was encoded in each
cell.  In this section, we show how the simulation of Section~\ref{sec:simulation} can
be modified to use only a single trigger, provided that infinitely long
gear chains are used.  We then sketch how these infinite chains can be
removed at the expense of computation time (more balls needed to pass
through the structure), by using a ``query ladder'' via which signals
from the tape cells may travel upward one cell at a time towards the
finite control at the top, as each ball passes through.

We replace the triggers T0 and T1, which are essentially being used to
communicate whether the symbol being scanned by the head is 0 or 1, with
a new gear bit ``$c$'' holding the contents of the scanned cell. The bit
$c$ is in every cell, and these are all connected together and join an
infinite chain that runs along the left side of the structure, reaching
the top where it can be queried just before the STM, so that a computing
ball may be routed into the correct half of the STM.

In the previous construction, we alternated computing balls with query
balls, so that the query ball can determine the value of the scanned
symbol (now held in bit $c$) in the new cell being scanned after each
computing ball executed a transition. We still need to do this, but
since we no longer have triggers to initiate the query ball or computing
ball sequence, we use a new mechanism employing a ``polling ball'', with
a new bit ``$r$'' to indicate when the machine is ready for either a new
computing ball or a new query ball.  Like $c$, the bit $r$ reaches into
every cell (which are now called ``5-cells''), and connects to an
infinite chain. The $r$-chain runs along the right side of the structure,
and reaches the top.

Polling balls are balls that are released periodically every 100 steps
or so, from a hopper at the top.  When the polling ball reaches $r$, it
branches. If $r = 0$, the machine is not ready to make progress as a
computation is still in progress. So, a new polling ball is released,
and the current ball is intercepted. If $r= 1$, then the ball in the
structure below has completed its mission, and, after setting $r$ to 0, the polling ball becomes
either a computing ball, or a query ball, depending on the state of a
toggle bit that alternates between the two events.  Refer to Figure~\ref{fig:polling}.

\begin{figure}[htp]
\begin{center}
\includegraphics[scale=0.35]{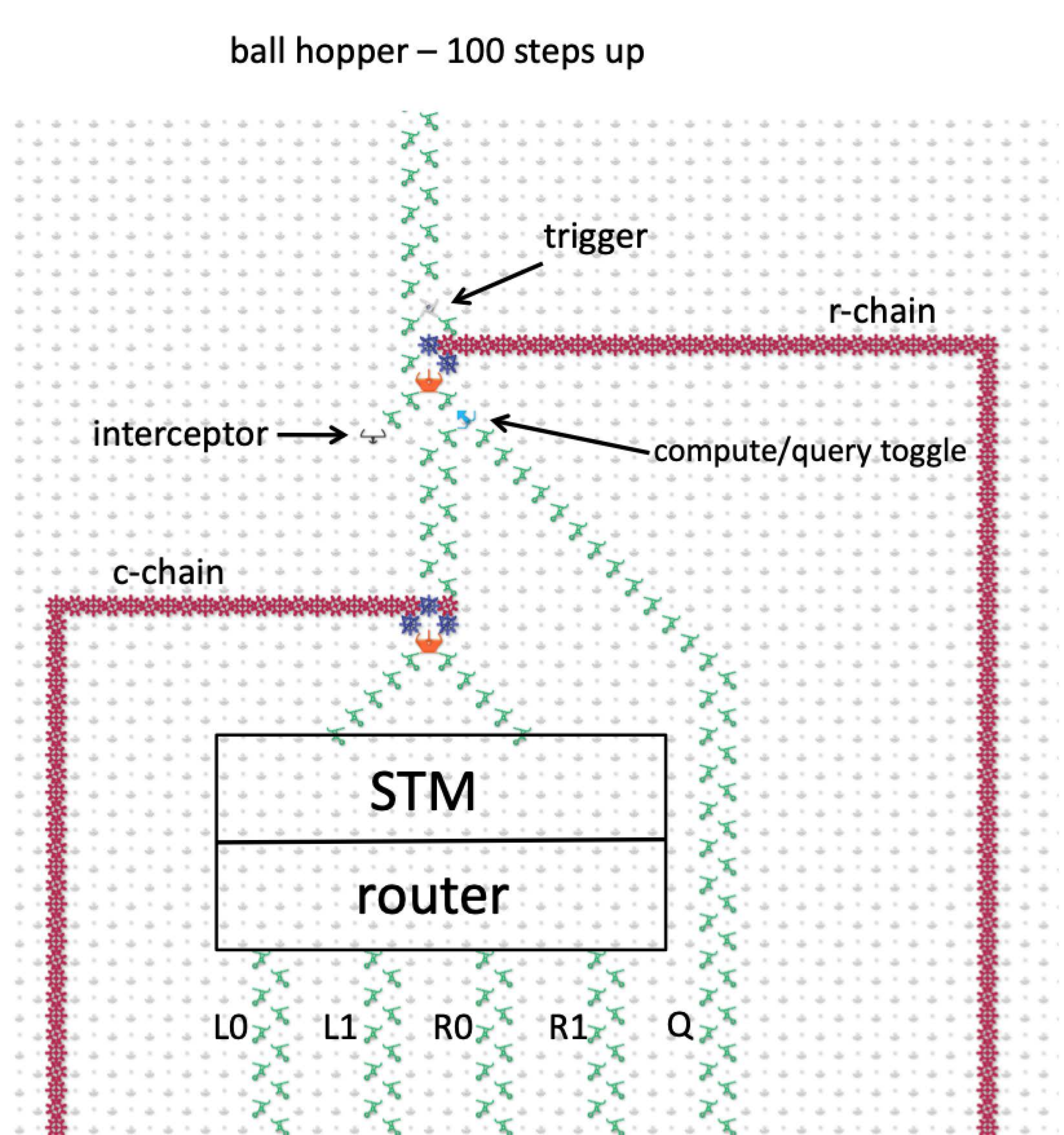}
\end{center}
\caption{Top of structure; polling balls may become computing balls
or query balls.}
 \label{fig:polling}
\end{figure}

If the toggle bit is in the 0 state, the ball becomes a Q-ball and is
routed down the Q line, and proceeds as in Section~\ref{sec:simulation}, but instead of
triggering T0 or T1 (which have been eliminated) after reading the $S$ bit
in the currently scanned cell, it sets $c$ to the value of $S$, sets $r$ to 1,
and is then intercepted.  The next polling ball will see $r=1$ again, and
become a computing ball.

If the compute/query toggle bit is in the 1 state, then the ball
becomes a computing ball and is routed to cross the $c$-chain where the
current symbol $c$ is tested and then the ball is routed into the correct
half of the STM.  From there, the behavior is identical to that
described in the construction in Section~\ref{sec:simulation}- in particular - the ball
travels down the Bus and when encountering $h=1$, sets it to 0, and
branches right into a 5-cell, and changes $S$ according to the
instruction. It passes $r$, setting it to 1, and then just as in the
construction of Section~\ref{sec:simulation}, if needed for a Left move, $p$ is set to 1 so
that the previous 5-cell's copy of $h$ is set to 1.  For a Right move, $p$
remains 0, and the ball goes down to set the next $h$ to 1. Thus in
either case, $r$ is set to 1 just before the ball is intercepted. Note
that at this point, $h$ has been changed to reflect that the TM is
scanning either the cell above or below, and the ``$c$'' bit incorrectly
holds the value of the cell that the head was just at.  The next
polling ball will become a Q-ball used to update the value~of~$c$. 

As mentioned, the infinite $c$-chains and $r$-chains reach into each
5-cell, occurring in the order $c$, then $r$, just below the $S$ bit and
before the $p$ bit. These will be some squiggly chains to be sure to
avoid intersecting with other gear chains, but notice in Figure~\ref{fig:architecture} that
the $S$ chains are short and do not connect to anything, so they can be
navigated around. And, there are many openings to the right (between $p$
in even numbered cells and $h$ in the next odd numbered cell) and to the
left (between odd $p$ and the next even $h$) for the $c$-chain and $r$-chain to
squeeze through.  Let's take a closer look.

\subsection*{Topology} 
A picture (Figure~\ref{fig:rctopology}) is worth the following words: In all odd 5-cells, the 
$r$-chain goes down around
the left side of $p$, then down and under $p$, heading out to the right
edge. The $c$-chain goes up around the right sides of $S$ then $h$, then
swings left, and out to the left edge. In even 5-cells, the $r$-chain
goes left, up around the left sides of $c, S,$ and $h$, then right and out
to the right edge.  The $c$-chain goes down, around the right edges of
 $r$ and $p$, then goes left and out to the left edge.  Thus, from each
5-cell, the $r$ and $c$ variables can reach the right and left edges,
respectively, and join an infinitely long chain leading up to the top.
~Once at the top, these long chains turn inward above the finite
control, with $r$ above the $c$. 

\begin{figure}[htp]
\begin{center}
\includegraphics[scale=0.35]{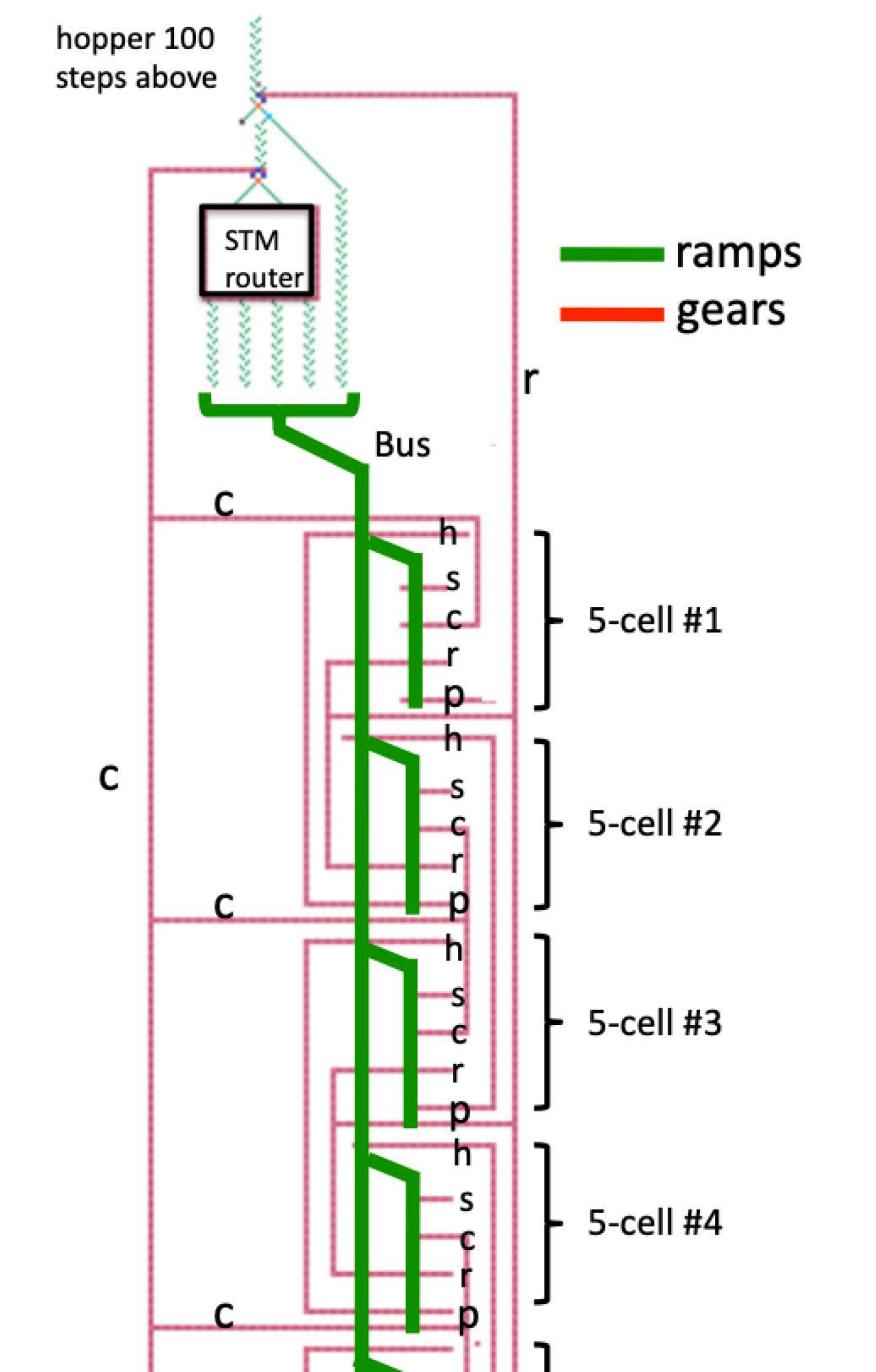}
\end{center}
\caption{Topology of $r$- and $c$- gear chains.}
 \label{fig:rctopology}
\end{figure}

Note that there are never two computing balls, or two query balls, or
one of each, active at the same time. At most, there are two balls
active: (1) a polling ball, and (2) either a computing or query ball
below.

A possible {\em race condition} occurs if at some point, a polling
ball near the top tries to cross the $r$-chain or $c$-chain at the exact
moment that a computing ball or query ball within a 5-cell below is
setting the $r$-chain or $c$-chain, respectively  The machine may jam. 
Even if the two balls are both executing an ``Ignore'' instead of a
write-0, it could still jam.  This seems like a fundamental problem
with more than one active ball if they both can access the same variable
at the same time. However, assuming balls fall at a constant rate, or at
least an ultimately periodic one, we can control the relative times that
balls are reading and writing the same variables by adjusting the
distances in various parts of the construction, thereby avoiding race
conditions.

For example, suppose that a 5-cell can be implemented within 100 vertical positions (``steps'')
on the board. By padding with extra space, we enforce a periodicity of
exactly 100 steps, and ensure that the polling ball reads and writes $r$
and $c$ near the beginning of the period, and the computing balls and
query balls read and write $r$ and $c$ near the middle of the period. 
~Specifically, from the ball hopper, we place a ramp of 100 steps to
reach the trigger, thus balls are released every 100 time steps.  Just
below the trigger, $r$ is read and branched on, and $c$ is read (if $r = 1$).
~ These can be made to occur at times 10 and 20 mod 100 (or earlier),
by positioning them at 110 and 120 steps below the ball hopper, (i.e.,
10 and 20 steps below the trigger, respectively.)  We position the
first 5-cell so that its ``$c$'' variable is at a position congruent to
50 mod 100.  We place the $r$ bit at a position 10 lower, hence at a
position congruent to 60 mod 100.  This leaves more than ample room
between the bits for branching.  By this construction, all polling
balls read $r$ and $c$ at times 110 and 120 mod 100, respectively, and all
computing balls or query balls (of which there is only one at any time)
within a 5-cell read and write $c$ and $r$ at times congruent to 50 and 60
mod 100.  Thus, no race condition can occur between a polling ball and
a computing or query ball.  If the implementation of a 5-cell needs
more than 100 steps, we just change the number 100 to a suitable larger
number.

The simulation begins with $Q_1 = r = 1, c =$ the symbol in the first
5-cell, the first $|w|$ 5-cells holding the symbols
comprising the input string $w$, and $h = 1$ in the first 5-cell. The
compute/query toggle bit is set to 1, indicating that the machine is
ready to compute.  All other gear bits are set to 0 (except for $p$ in
the second 5-cell, which is geared to $h$ in the first, hence set to 1). 
The trigger is pushed, and the computation begins with the first ball
ultimately being routed into the STM.\\

This completes the description of the modified machine.
The proof that the simulation is correct relies on a number of lemmas,
closely related to those previously shown. 

Lemma~\ref{lem:oneQ} (stating that at any time exactly one state variable $Q_i$ is set to 1) is used.   We need the following slight modification of Lemma~\ref{lem:oneh}: 

\begin{lemma}
\label{lem:still-oneh}
After each computing ball is intercepted, there is
exactly one $h$ set to 1.
\end{lemma}

\begin{lemma}
\label{lem:still-alternate}
Query balls and computing balls alternate, interspersed
with any number of polling balls.
\end{lemma}

\begin{lemma}
\label{lem:computing-continues}
Every computing ball either sets the state to the halt state
$Q_n$, or else is ultimately intercepted and eventually followed by another
computing ball unless the TM has moved off of the left edge of the
tape.
\end{lemma}

The proofs of {Lemmas~\ref{lem:still-oneh} and~\ref{lem:still-alternate}
are similar to the proofs of Lemmas~\ref{lem:oneh} and~\ref{lem:alternating}, respectively.  Lemma~\ref{lem:computing-continues} is straightforward.
The following additional lemmas are proved via induction, and
are not difficult. Most }{follow}{~by construction.

\begin{lemma}
\label{lem:correct-Q}
Just after the nth computing ball is intercepted, the unique
$Q_i$ set to 1 is the state that the TM would be in after n steps of
computation.
\end{lemma}

\begin{lemma}
\label{lem:correct-head}
Just after the nth computing ball is intercepted, the 5-cell
that contains the unique $h$ set to 1 is the one that corresponds to the
cell being scanned by the TM after n steps of its computation.
\end{lemma}

\begin{lemma}
\label{lem:correct-tape}
Just after the nth computing ball is intercepted, the
settings of the $S$ variables in the 5-cells accurately reflect the
contents of the TM tape after n steps of computation.
\end{lemma}

\begin{lemma}
\label{lem:correct-c}
Unless $Q_i$ = $Q_n$, after the nth query ball is intercepted, the
setting of the variable $c$ accurately reflects the symbol in the cell
being scanned after the nth step of computation of the TM, and this
remains true while the next computing ball travels through the finite
control.
\end{lemma}

\begin{lemma}
\label{lem:correct-r}
After any query ball or computing ball is intercepted, the
variable $r$ is set to 1, unless the computation has finished.
\end{lemma}

\begin{theorem}
\label{thm:correct-simulation}
For all $n$, just after the $n$th computing ball is
intercepted, the TT structure accurately reflects n steps of the
computation of the TM with the given input. 
\end{theorem}

\begin{proof}
We need to show that (a), the structure reflects that the
tape head is in the correct position, (b) the tape is accurate, and (c)
the state is accurate.  Lemmas~\ref{lem:still-oneh} and \ref{lem:correct-head}
ensure that (a) holds. By
Lemma~\ref{lem:correct-tape}, (b) holds. And, by Lemmas~\ref{lem:oneQ} and \ref{lem:correct-Q}, (c) is true.  (Note Lemma~\ref{lem:correct-Q}
relies on Lemma~\ref{lem:correct-c}).
\qed
\end{proof}

Corollary~\ref{cor:simulates}, which states that the structure computes exactly what the TM does, still holds, but with a slightly different proof as the TM
doesn't alternate query balls and compute balls, but there are
intervening polling balls. (We rely on Lemmas~\ref{lem:still-alternate},
\ref{lem:computing-continues}, 
and \ref{lem:correct-r}.)

\subsection*{Removing the Infinite Gear Chains}

One might object to the inclusion of infinite gear chains, finding them
an unreasonable extension of the model.  However, if desired, the
infinite $r$-chain and $c$-chain can both be replaced with a ``ladder'' of
nested $c$'s and $r$'s, as shown. This structure is called the ``query
ladder''.  This is actually a general method that can be used to
replace any number of infinite vertical chains with non-infinite nested
alternating ``C-shapes''. The price paid is that to move the signal up
the ladder $n$ steps requires $n$ balls to pass through.

How does it work? Referring to Figure~\ref{fig:one-rung}, a ball falls down the left
edge, where the variable $r$ is tested. As long as $r=0$, it continues
straight down through the ``rungs'' of the ladder, ignoring ``$c$'' values
as it passes through them.  When $r=1$ is found, the ball branches to the
right while setting $r=0$ on this rung, tests $c$ and branches, and then
writes the value of $c$, and then $r=1$, which are connected to the $r$-test
and $c$-test on the rung above. The ball is then intercepted.  The next
ball to pass down the ladder will find $r=1$ and the current value of $c$ at
one higher rung on the ladder.  Note that there will only ever be at
most one rung with $r=1$. 

\begin{figure}[htp]
\begin{center}
\includegraphics[scale=0.35]{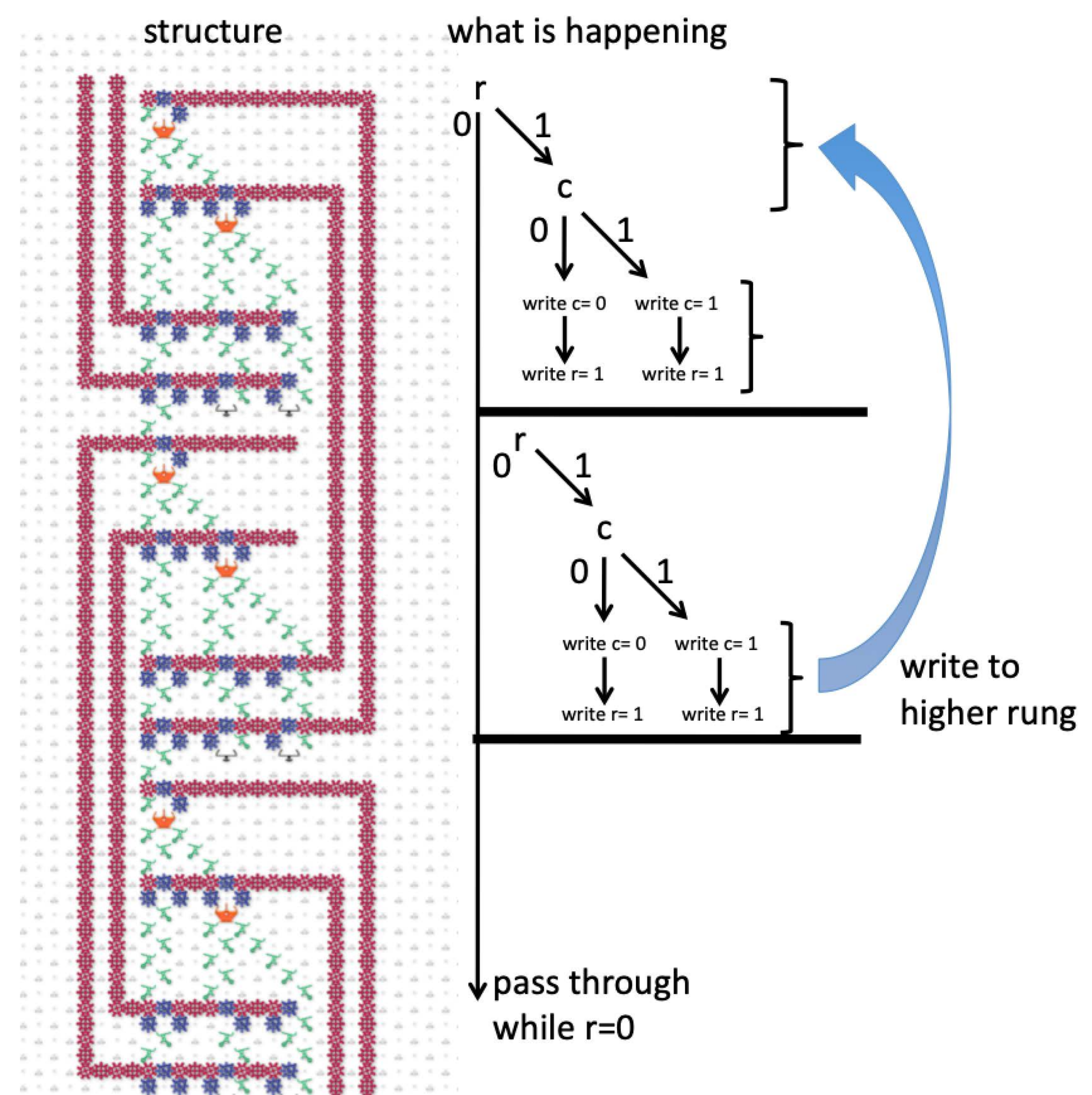}
\end{center}
\caption{A Signal Moves Up One Rung on the Query Ladder.}
 \label{fig:one-rung}
\end{figure}

In order to make use of the ladder, which replaces the infinite $c$-chain
and $r$-chain, we must be able to connect the values of $c$ and $r$ from each
5-cell to this ladder without crossing gear chains.  Instead of routing
the $c$ bits to the left and $r$ bits to the right, as in Figure~\ref{fig:rctopology}, we
route both $c$ and $r$ bits to the right, where the values join the ladder
as shown in Figure~\ref{fig:connections} 
 below, which omits all pieces except gear chains
due to space limitations. As can be seen, this is not difficult to do
with a little care.  Between each gap between $p$ in an odd cell numbered $k$
and $h$ in the even cell numbered $k+1$ just below it, we route $c$ and $r$ from both
cells, and the order of the values as they come out will be, top to
bottom: $r$ (cell $k$), $c$ (cell $k$), $r$ (cell $k+1$), $c$ (cell $k+1$).

Now these chains must be connected to the ladder, as follows: The first $r$ and $c$ pass through the gap between left
braces and then connect to the bottom of the left brace they just went
under, and the second $r$ and $c$ pass through the gap and head up to
connect to the $r$ and $c$ from the top of the right brace. Refer to Figure~\ref{fig:connections}.    Note that no pieces other than gear chains are shown due to space limitations.  In the actual structure, for example, spacing
would be farther apart to allow ramp chains to pass through the gear chains in many places, using the ``ignore'' structure.

\begin{figure}[htp]
\begin{center}
\includegraphics[scale=0.25]{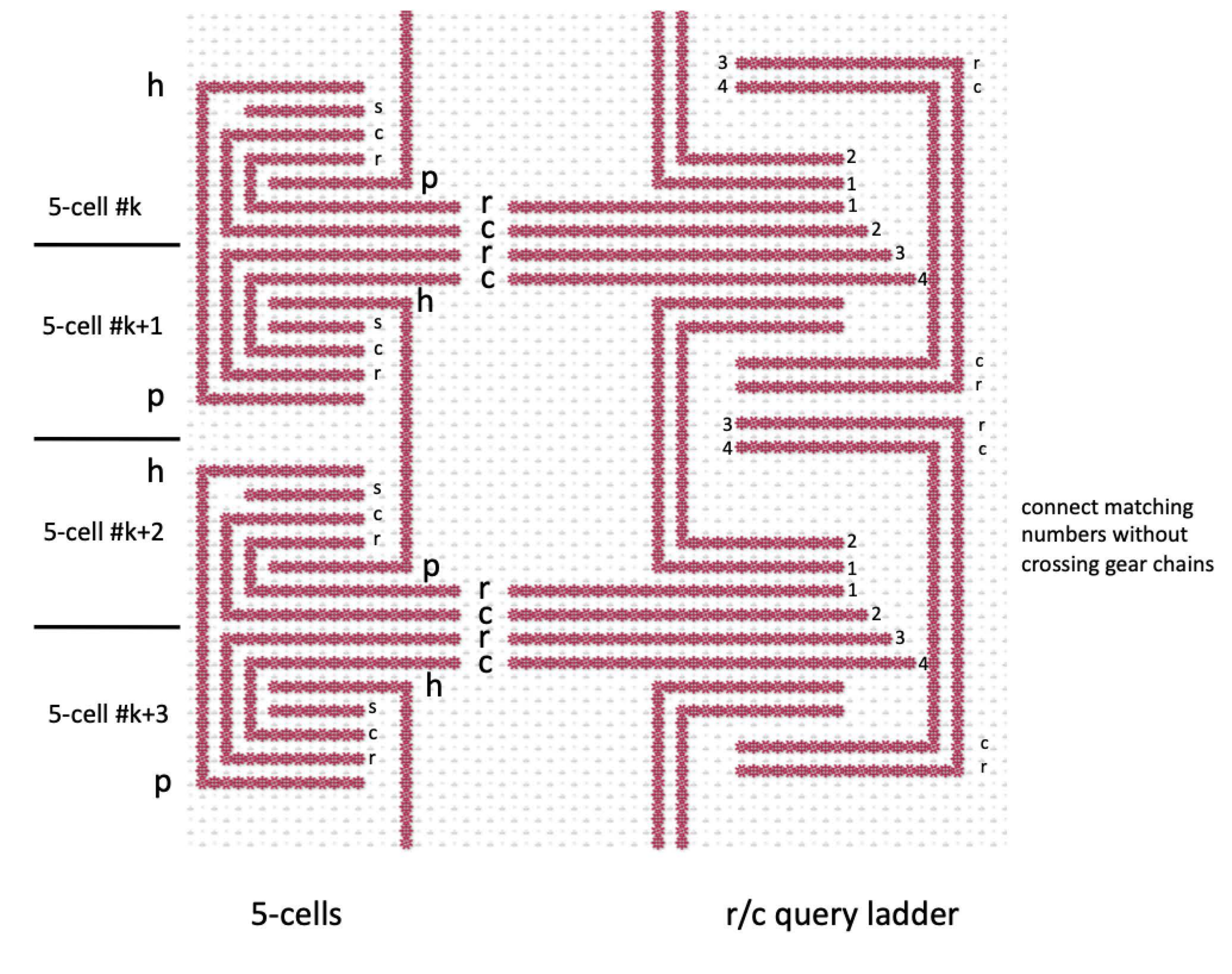}
\end{center}
\caption{Connections between 5-cells and query ladder.}
 \label{fig:connections}
\end{figure}

We sketch the ideas needed to make use of the query ladder; the
full construction and argument are left as an exercise.   There are
three modes of ball behavior: 
\begin{enumerate}
\item Computing within the 5-cells. In
this mode a ball is searching for the active cell with $h=1$, and will execute
the transition when it is found. 
\item Querying the currently scanned
symbol. In this mode a ball is traversing down looking for the scanned
symbol $S$ in the new cell for which $h=1$. It then sets $c$ to $S$, and $r$ to 1, causing the query ladder to have the correct values loaded.
\item
Climbing. In this mode balls repeatedly traverse down the query ladder
to cause the $r$-bit and $c$-bit to climb one rung at a time.
\end{enumerate}

It is not difficult to create these three modes as cyclic behavior for
balls being released at the top, by having a 2-bit, 4-state counter with
one null state.  The counter does not increase while climbing balls
are being released, until the $r=1$ signal is received at the top of the
ladder.

Balls are released every 100 steps, and query the value of $r$ at the
top. If $r=1$, the ball becomes a computing ball. The next ball becomes
a querying ball. After the computing ball is intercepted in the active
5-cell, the querying ball comes along (about 100 steps behind it), reads
the symbol $S$ in the newly scanned cell (where $h=1$), it sets $r$ to 1, and
sets $c$ to the value $S$,  and these values are therefore written to the
query ladder nearby.    All balls released after the querying ball
are climbing balls, which move $r$ and $c$ up the ladder.   When $r=1$ is
seen by a newly released ball at the top, $r$ is reset to 0 and the mode
changes to computing ball, to begin the cycle again.

\section{Conclusion}

We have seen three constructions of increasing complexity that give
direct simulations of a Turing machine by natural extensions of the
Turing Tumble toy, using an infinite playing board and an infinite
number of pieces. In the final construction, only one trigger and ball
hopper are used, and no infinite chains are needed.   As discussed in the Introduction, the infinite board and pieces are initially set with a simple
repeating pattern of 0 bits, to represent the tape of the TM.  If
needed, the simulation could proceed without an infinite board, provided
that whenever needed, an additional board and additional pieces are
added.

It is interesting to note that because of the existence of
{\href{https://en.wikipedia.org/wiki/Universal_Turing_machine}{Universal
TMs,}  there is a single TM (for example, one with only fifteen states and two
symbols) that is universal. Hence, a single fifteen-state TM can be encoded
on Turing Tumble that can simulate any computation.}

\section*{Acknowledgments}
 
We thank El$'\!$endia Starman, Jesse Crossen, Richard
Pawson, and Teijiro Isokawa for helpful discussions.  
The images were produced via screenshots from the
TT {\href{https://tinyurl.com/ttsimulator}{simulator}
of J. Crossen.    Finally, we thank Paul and Alyssa Boswell, for their great invention and for making so many resources readily available. 


\bibliography{TT-is-TC}

%


\end{document}